\documentclass[12pt,draftclsnofoot,onecolumn]{IEEEtran}
\usepackage{epsfig,latexsym,amsmath,amssymb,setspace,cite,color,graphicx,algorithm,algorithmic,cases}
\usepackage{amssymb,amsmath,amsfonts,bm,epsfig,graphicx,latexsym}
\usepackage{rotating,setspace,latexsym,epsf,color}
\usepackage{cite,authblk,epstopdf,footnote}
\usepackage{float}
\usepackage{tabu}
\usepackage{array}
\usepackage[subrefformat=parens,labelformat=parens]{subfig}
\usepackage{bbm}
\usepackage{amsthm}
\usepackage{thmtools}
\usepackage{setspace}
\usepackage[inline]{enumitem} 
\usepackage{algorithm,algorithmic}

\usepackage{eqparbox}

\declaretheorem[style=plain,qed=$\blacksquare$]{theorem}
\declaretheorem[style=plain,name=Definition,qed=$\blacksquare$]{Definition}
\declaretheorem[style=plain,name=Lemma]{lemma}
\declaretheorem[style=plain,name=Remark,qed=$\blacksquare$]{remark}
\declaretheorem[style=plain,name=Example,qed=$\blacksquare$]{example}
\declaretheorem[style=plain,name=Proposition,qed=$\blacksquare$]{proposition}

\allowdisplaybreaks

\def\mc{\ensuremath\mathcal}

\newcounter{opt_ct}
\stepcounter{opt_ct}

\begin{document}

%\sloppy

%% Paper Title
%% You can use linebreaks \\ within to get better formatting as
%% desired. 
\title{Coded Caching for Heterogeneous Systems: {A}n Optimization Perspective\thanks{ This work was presented in part at the IEEE Wireless Communications and Networking Conference (WCNC), San Francisco, CA, 2017, and the IEEE International Conference on Communications (ICC), Paris, France, 2017. This work was supported in part by NSF Grants 1526165 and 1749665.}}
% 1422347

\author{
\IEEEauthorblockN{Abdelrahman M. Ibrahim, Ahmed A. Zewail and Aylin Yener}\\
  %\IEEEauthorblockN{Abdelrahman M. Ibrahim \qquad Ahmed A. Zewail \qquad Aylin Yener}\\
    \IEEEauthorblockA{Wireless Communications and Networking Laboratory (WCAN)\\
    The School of Electrical Engineering and Computer Science\\
    The Pennsylvania State University, University Park, PA 16802\\
   \textit{ami137@psu.edu \ \ \ zewail@psu.edu \ \ \ yener@ee.psu.edu}
}
}
\maketitle
%\author{Abdelrahman M. Ibrahim}
%\author{Ahmed A. Zewail}
%\author{Aylin Yener}
%\affil{\normalsize Wireless Communications and Networking Laboratory (WCAN)\\
%Electrical Engineering Department\\
%The Pennsylvania State University, University Park, PA 16802.\\
%\em \{ami137,aiz103\}@psu.edu \qquad yener@engr.psu.edu}
\vspace{-0.9in}
\begin{center}
%\today
\end{center}
\vspace{-0.1in}
\doublespacing
\begin{abstract}
\vspace{-0.1in}

In cache-aided networks, the server populates the cache memories at the users during low-traffic periods, in order to reduce the delivery load during peak-traffic hours. In turn, there exists a fundamental trade-off between the delivery load on the server and the cache sizes at the users. In this paper, we study this trade-off in a multicast network where the server is connected to users with unequal cache sizes and the number of users is less than or equal to the number of library files. We propose centralized uncoded placement and linear delivery schemes which are optimized by solving a linear program. Additionally, we derive a lower bound on the delivery memory trade-off with uncoded placement that accounts for the heterogeneity in cache sizes. We explicitly characterize this trade-off for the case of three end-users, as well as an arbitrary number of end-users when the total memory size at the users is small, and when it is large. Next, we consider a system where the server is connected to the users via rate limited links of different capacities and the server assigns the users' cache sizes subject to a total cache budget. We characterize the optimal cache sizes that minimize the delivery completion time with uncoded placement and linear delivery. In particular, the optimal memory allocation balances between assigning larger cache sizes to users with low capacity links and uniform memory allocation.

\end{abstract}
\vspace{-0.05in}
\begin{IEEEkeywords}
\vspace{-0.1in}
Coded caching, uncoded placement, cache size optimization, multicast networks. 
\end{IEEEkeywords}
\vspace{-0.1in}

%\footnotetext[1]{This material is based upon work supported by the Marie Curie International Research Staff Exchange Scheme Fellowship PIRSES-GA-2010-269132 AGILENet within the 7th European Community Framework Program.}
%\vspace{-0.12in}
% ------------------------------------------------------------------------
% ------------------------------------------------------------------------

\newpage
\section{Introduction}
The immense growth in wireless data traffic is driven by video-on-demand services, which are expected to account for $82\%$ of all consumer Internet traffic by $2020$ \cite{cisco}. The high temporal variation in video traffic leads to under-utilization of network resources during off-peak hours and congestion in peak hours \cite{almeroth1996use}. Caching improves uniformization of network utilization, by pushing data into the cache memories at the network edge during off-peak hours, which in turn reduces congestion during peak hours. The seminal work \cite{maddah2014fundamental} has proposed a novel caching technique for a downlink setting, in which a server jointly designs the content to be placed during off-peak hours and the delivered during peak hours, in order to ensure that multiple end-users can benefit from delivery transmissions simultaneously. These multicast coding opportunities are shown to provide gains beyond  local caching gains, which result from the availability of a fraction of the requested file at the user's local cache. They are termed global caching gains since they scale with the network size. Reference \cite{maddah2014fundamental} has shown that there exists a fundamental trade-off between the delivery load on the server and the users' cache sizes. 

The characterization of this trade-off has been the focus of extensive recent efforts  \cite{wan2016optimality,wan2017novel,yu2016exact,chen2014fundamental,amiri2017fundamental,
gomez2016fundamental,ghasemi2017improved,lim2017information,tian2017uncoded,
wang2017improved,yu2017characterizing}. In particular, references \cite{wan2016optimality,wan2017novel,yu2016exact} have characterized the delivery load memory trade-off with the uncoded placement assumption, i.e., assuming that the users cache only uncoded pieces of the files. The delivery load memory trade-off with general (coded) placement has been studied in \cite{chen2014fundamental,amiri2017fundamental,
gomez2016fundamental,ghasemi2017improved,lim2017information,tian2017uncoded,
wang2017improved,yu2017characterizing}. Coded caching schemes were developed for various cache-aided network architectures, such as multi-hop  \cite{ji2015comb,zewail2017combination,
wan2017caching}, device-to-device (D2D) \cite{ji2016fundamental,ibrahim2018device}, multi-server \cite{shariatpanahi2016multi}, lossy broadcast \cite{timo2015joint,ghorbel2016content,bidokhti2017benefits,amiri2018cache}, and interference networks\cite{naderializadeh2017fundamental,xu2017fundamental}. In addition to network topology, several practical considerations have been studied, such as the time-varying nature of the number of users \cite{maddah2015decentralized}, distortion requirements at the users \cite{yang2018coded,hassanzadeh2015distortion,ibrahim2018distortion}, non-uniform content distribution \cite{niesen2017coded,ji2017order,
ramakrishnan2015efficient,zhang2015coded,jin2017structural}, delay-sensitive content \cite{niesen2015delay}, and systems with security requirements \cite{ravindrakumar2018private,sengupta2015fundamental,zewail2016fundamental}. 

End-users in practical caching networks have varying storage capabilities. In this work, we address this system constraint by allowing the users to have distinct cache sizes. In particular, we study the impact of heterogeneity in cache sizes on the delivery load memory trade-off with uncoded placement. Models with similar traits have been studied in references\cite{yang2018coded,wang2015fundamental,amiri2017decentralized,sengupta2016layered}. In particular, references \cite{wang2015fundamental,amiri2017decentralized} have extended the decentralized caching scheme in \cite{maddah2015decentralized} to systems with unequal cache sizes. References \cite{yang2018coded,sengupta2016layered} have proposed a centralized scheme in which the system is decomposed into layers such that the users in each layer have equal cache size. More specifically, the scheme in \cite{maddah2014fundamental} is applied on each layer and the optimal fraction of the file delivered in each layer is optimized. Additionally, reference \cite{sengupta2016layered} has proposed grouping the users before applying the layered scheme which requires solving a combinatorial problem. In a follow-up work to some of our preliminary results presented in \cite{ibrahim2017centralized}, reference \cite{daniel2017optimization} proposed optimizing over uncoded placement schemes assuming the delivery scheme in \cite{maddah2015decentralized}.

In this work, we focus on uncoded placement and linear delivery, where the server places uncoded pieces of the files at the users' cache memories, and the multicast signals are formed using linear codes. Our proposed caching scheme outperforms the schemes in \cite{yang2018coded,sengupta2016layered,daniel2017optimization}, because it allows flexible utilization of the side-information in the creation of the multicast signals, i.e., the side-information stored exclusively at $t$ users is not restricted to multicast signals of size $t+1$ as in \cite{maddah2014fundamental,yang2018coded,wang2015fundamental
,amiri2017decentralized,sengupta2016layered,daniel2017optimization}. We show that the worst-case delivery load is minimized by solving a linear program over the parameters of the proposed caching scheme. In order to evaluate the performance of our caching scheme, we derive a lower bound on the worst-case delivery load with uncoded placement. Using this bound, we explicitly characterize the delivery load memory trade-off for arbitrary number of users with uncoded placement in the small total memory regime, large total memory regime, the definitions of which are made precise in the paper, and for any memory regime for the instance of three-users. Furthermore, we compare the achievable delivery load with the proposed lower bound with uncoded placement, and the lower bounds with general placement in \cite{wang2017improved,amiri2017decentralized}. From the numerical results, we observe that our achievable delivery load coincides with the uncoded placement bound.

Next, inspired by the schemes developed for distinct cache sizes we consider a middle ground between noiseless setups \cite{maddah2014fundamental,yang2018coded,sengupta2016layered} and noisy broadcast channels with cache-aided receivers \cite{timo2015joint,ghorbel2016content,
bidokhti2017benefits,amiri2018cache}. More specifically, we assume that the server is connected to the users via rate limited links of different capacities, and the server assigns the users' cache sizes subject to a cache memory budget. Reference \cite{tang2017coded} has considered a similar model and proposed jointly designing the caching and modulation schemes. Different from \cite{timo2015joint,ghorbel2016content,
bidokhti2017benefits,amiri2018cache,tang2017coded}, we consider a separation approach where the caching scheme and the physical layer transmission scheme are designed separately. This is inline in general with the approach of \cite{maddah2014fundamental} and followup works that consider server to end-users links as bit pipes. We focus on the joint optimization of the users' cache sizes and the caching scheme in order to minimize the worst-case delivery completion time. More specifically, the optimal memory allocation, uncoded placement, and linear delivery schemes are again obtained by solving a linear program. For the case where the cache memory budget is less than or equal to the library size at the server, we derive closed form expressions for the optimal memory allocation and caching scheme. We observe that the optimal solution balances between assigning larger cache memories to users with low capacity links, delivering fewer bits to them, and uniform memory allocation, which maximizes the multicast gain.

%\subsection{Organization}
%The remainder of this paper is organized as follows. In the next section, we describe the system model and the main assumptions. In section \ref{sec_motv}, we motivate our caching scheme by an example. The cache placement and delivery schemes are developed in Sections \ref{sec_cach_plac} and \ref{sec_cach_dlv}, respectively. In Section \ref{sec_cach_opt}, we formulate the optimization problem and present our results on the trade-off between the delivery load and the users' cache sizes. In Section \ref{sec_mem_opt}, we study the  joint optimization of caching scheme and memory allocation in systems with rate limited links of different capacities and a cache memory budget. The numerical results are presented in Section \ref{sec_numerical}. Finally, we draw our conclusions in Section \ref{sec_con}.

%########################################################################################
%########################################################################################
%\begin{figure}[t]
%\includegraphics[scale=1.5]{figures/system_model.eps}
%\centering
%\caption{Centralized caching system with unequal cache sizes.}\label{fig_sys_model}
%%\vspace{-.1 in}
%\end{figure}

\begin{figure*}[t]
	\centering
	\begin{tabular}{cc}
	\hspace{-0.2in} \subfloat[Fixed cache sizes and equal download rates.]{
				\label{fig_sys_model}
				\includegraphics[scale=1.5]{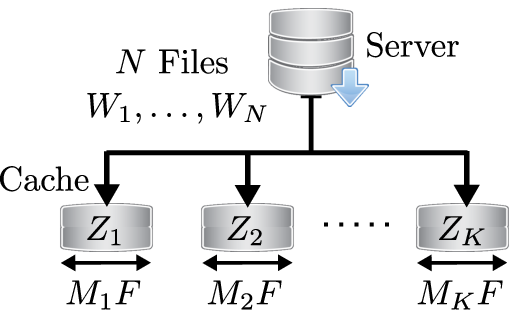} }
		
		& \hspace{-0.1in}
		\subfloat[Cache memory budget and unequal download rates.]{ \label{fig_sys_model_gen}
			\includegraphics[scale=1.5]{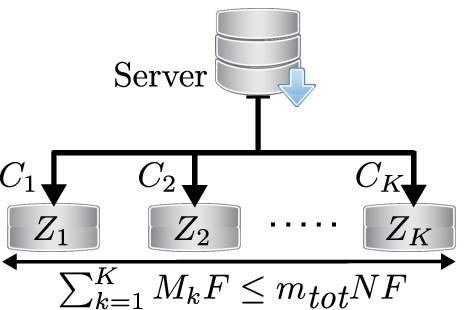}}
		\\
	\end{tabular}       
	\caption{Centralized caching system with unequal cache sizes.}
\vspace{-0.2in}
\end{figure*} 
\vspace{-0.1in}
\section{System Model}\label{sec_sys_mod}
\textit{Notation:} Throughout the paper, vectors are represented by boldface letters, sets of policies are represented by calligraphic letters, e.g., $\mathfrak{A} $, $ \oplus$ refers to bitwise XOR operation, $(x)^{+}\!  \triangleq \! \max \{0,x\} $, $|W|$ denotes the size of $W$, $\mc A \setminus \mc B $ denotes the set of elements in $\mc A$ and not in $\mc B $, $ \phi$ denotes the empty set, $[K] \triangleq \{1,\dots,K\}$, $\mc A \subset \mc B $ denotes $\mc A$ being a subset of or equal to $\mc B$, $\subsetneq_{\phi} [K]$ denotes non-empty subsets of $[K]$, and $\mc P_{\mc A}$ is the set of all permutations of the elements in the set $\mc A$, e.g., $\mc P_{\{1,2\}}= \{[1,2], \ [2,1]\}$.

Consider a centralized system consisting of a server connected to $K$ users via an error-free multicast link \cite{maddah2014fundamental}, see Fig. \subref*{fig_sys_model}. A library $\{ W_{1}, \dots, W_{N}\}$ of $N$ files, each with size $F$ bits, is stored at the server. User $k$ is equipped with a cache memory of size $M_k F$ bits. Without loss of generality, we assume that $M_1 \leq M_2 \leq \dots \leq M_K$. We define $m_k=M_k/N$ to denote the memory size of user $k$ normalized by the library size $N F$, i.e., $m_k \in [0,1]$ for $M_k \in [0,N]$. The cache size vector is denoted by $\bm M =[M_1,\dots,M_K]$ and its normalized version by $\bm m =[m_1,\dots,m_K]$. We focus on the case where the number of files is larger than or equal to the number of users, i.e., $N \geq K $. 

In Section \ref{sec_mem_opt}, we introduce rate limited download links of distinct capacities to the model. In particular, we consider that the link between the server and user $k$ has capacity $C_k$ bits per channel use, which we refer to as the \textit{download rate} at user $k$, as illustrated in Fig. \subref*{fig_sys_model_gen}. We denote the collection of link capacities by $\bm C =[C_1,\dots,C_K]$. In this setup, we seek the system configuration with best performance, including the memory sizes $\{ M_k\}$, subject to $\sum_{k=1}^{K} M_k F \leq  m_{\text{tot}} N F$ bits, where $m_{\text{tot}}$ is the cache memory budget normalized by the library size.

The system operates over two phases: placement phase and delivery phase. In the placement phase, the server populates users' cache memories without knowing the users' demands. User $k$ stores $Z_k$, subject to its cache size constraint, i.e., $|Z_k| \leq M_k F$ bits. Formally, the users' cache contents are defined as follows.
\begin{Definition}(Cache placement) A cache placement function $\varphi_k: [2^F]^N\rightarrow [2^{\lfloor M_kF \rfloor}]$ maps the files in the library to the cache of user $k$, i.e., $ Z_k = \varphi_k(W_1, W_2,..,W_N) $. 
\end{Definition}
%
%\begin{figure}[t]
%\includegraphics[scale=1.6]{figures/cache_model.eps}
%\centering
%\caption{Centralized caching system with rate limited download links.}\label{fig_sys_model_gen}
%%\vspace{-.1 in}
%\end{figure}
%
In the delivery phase, user $k$ requests file $W_{d_k}$ from the server. Users' demand vector $\bm d=[d_1, \dots, d_K]$ consists of independent uniform random variables over the files as in \cite{maddah2014fundamental}. In order to deliver the requested files, the server transmits a sequence of unicast/multicast signals, $X_{\mc T, \bm d}$, to the users in the sets $\mc T \subsetneq_{\phi} [K]$. $X_{\mc T, \bm d}$ has length $v_{\mc T} F$ bits, and is defined as follows.
\begin{Definition}(Encoding) Given $\bm d$, an encoding function $\psi_{\mc T, \bm d}: [2^F]^{K} \rightarrow [2^{\lfloor v_{\mc T}F \rfloor}] $ maps requested files to a signal with length $v_{\mc T} F$ bits, sent to users in $\mc T$, i.e., $X_{\mc T,\bm d}= \psi_{\mc T, \bm d}(W_{d_{1}},..,W_{d_{K}})$.
\end{Definition}
At the end of the delivery phase, user $k$ must be able to reconstruct $W_{d_k}$ from the transmitted signals $X_{\mc T, \bm d}, \mc T \subsetneq_{\phi} [K]$ and its cache content $Z_k$, with negligible probability of error.
\begin{Definition}(Decoding) A decoding function $\mu_{\bm d, k}: [2^{ \lfloor RF \rfloor}] \times [2^{ \lfloor M_k F \rfloor }] \rightarrow [2^F]$, with $R \triangleq \sum\limits_{\mc T \subsetneq_{\phi} [K]} v_{\mc T}$, maps cache content of user $k$, $Z_k$, and the signals $X_{\mc T, \bm d}, \mc T \subsetneq_{\phi} [K]$ to $\hat W_{d_k}$, i.e., $\hat W_{d_k} = \mu_{\bm d, k}\left(X_{\{1\}, \bm d},X_{\{2\}, \bm d},\dots,X_{[K], \bm d}, Z_k \right)$. 
\end{Definition}
A caching scheme is defined by $(\varphi_k(.),\psi_{\mc T, \bm d}(.),\mu_{\bm d, k}(.))$. The performance is measured in terms of the achievable delivery load, which represents the amount of data transmitted by the server in order to deliver the requested files. 
\begin{Definition} 
For a given normalized cache size vector $\bm m$, the delivery load $R(\bm m)$ is said to be achievable if for every $\epsilon > 0$ and large enough $F$, there exists $(\varphi_k(.),\psi_{\mc T, \bm d}(.),\mu_{\bm d, k}(.))$ such that $\max\limits_{\bm d, k \in [K]} Pr(\hat W_{d_k}\neq W_{d_k})\leq \epsilon$, and $R^*(\bm m) \triangleq \inf \{ R : R(\bm m) \text{ is achievable} \}$. 
\end{Definition}%

The set of cache placement policies $\mathfrak{A}$ considered in this work are the so-called uncoded policies, i.e., only pieces of individual files are placed in the cache memories. Since we have uniform demands, the cache memory at each user $k$ is divided equally over the files, i.e., $m_k F$ bits per file. We consider the set of delivery schemes $\mathfrak{D}$, in which multicast signals are formed using linear codes. The worst-case delivery load achieved by a caching scheme in $(\mathfrak{A},\mathfrak{D})$ is defined as follows.
\vspace{-0.15in} 
\begin{Definition} With placement and delivery policies in $ \mathfrak A$ and $\mathfrak{D}$, the worst-case delivery load is defined as $R_{\mathfrak A, \mathfrak D} \! \triangleq \! \max\limits_{\bm d} R_{\bm d,\mathfrak A, \mathfrak D} \! =  \! \! \!  \! \sum\limits_{\mc T \subsetneq_{\phi} [K]} \! v_{\mc T}$, and the minimum delivery load over all $R_{\mathfrak A, \mathfrak D}$ is denoted by $R^*_{\mathfrak A,\mathfrak D}(\bm m) \triangleq \inf \{ R_{\mathfrak A,\mathfrak D} : R_{\mathfrak A, \mathfrak D}(\bm m) \text{ is achievable} \}$.
%Clearly, we have $R^*_{\mathfrak A, \mathfrak D}   \triangleq  \inf\limits_{\mathfrak A, \mathfrak D} R_{\mathfrak A, \mathfrak D}  \geq  R^*_{\mathfrak A}  \geq  R^*$. 
\end{Definition}
\vspace{-0.22in}
\begin{Definition} The minimum delivery load achievable with a placement policy in $ \mathfrak A$ and any delivery scheme, is defined as $R_{\mathfrak A}^*(\bm m) \triangleq \inf \{ R_{\mathfrak A} : R_{\mathfrak A}(\bm m) \text{ is achievable} \}.$ 
\end{Definition}
\begin{remark}
Note that $R_{\mathfrak A, \mathfrak D}^* \geq R^*_{\mathfrak A} \geq R^* $, since $R^*$ is obtained by taking the infimum over all achievable delivery loads, $R^*_{\mathfrak A}$ is restricted to uncoded placement policies in $\mathfrak{A}$, and $R_{\mathfrak A, \mathfrak D}^*$ is restricted to cache placement and delivery policies in $\mathfrak{A}$ and $\mathfrak{D}$, respectively.
\end{remark}

%  Additionally, by taking the infimum over placement policies in $ \mathfrak A$ and all possible delivery policies, we get $R^*_{\mathfrak A}$.
% $R^*_{\mathfrak A,\mathfrak D}$ denotes the minimum delivery load achievable with a caching scheme in $(\mathfrak A,\mathfrak D)$. 

In Section \ref{sec_mem_opt}, we consider download links with limited and unequal capacities. Thus, $X_{\mc T, \bm d}$ will need to have a rate $ \leq \min\limits_{j \in \mc T} C_j$ \cite{dana2006capacity}. Additionally, there is no guarantee that the users outside the set $\mc T $ can decode $X_{\mc T, \bm d}$, as their download rates may be lower than $\min\limits_{j \in \mc T} C_j$. Consequently, a more relevant system-wide metric is the total time needed by the server to deliver all the requested files to all the users, defined as follows, assuming uncoded placement and linear delivery.

\begin{Definition}
With a placement policy in $ \mathfrak A$, and a delivery policy in $\mathfrak{D}$, the worst-case delivery completion time (DCT) is defined as $\Theta_{\mathfrak A, \mathfrak D} \triangleq \max\limits_{\bm d} \Theta_{\bm d,\mathfrak A, \mathfrak D} = \sum\limits_{\mc T \subsetneq_{\phi} [K]} \dfrac{v_{\mc T}}{\min\limits_{j \in \mc T} C_j}$.
\end{Definition}
Observe that, for $C_k=1, \forall k \in [K]$, $\Theta_{\mathfrak A, \mathfrak D} = R_{\mathfrak A, \mathfrak D} $.
%--------------------------------------------------------------------------------------
%--------------------------------------------------------------------------------------

\section{Motivational Example}\label{sec_motv}
In order to motivate our caching scheme which is tailored to capitalize on multicast opportunities to the fullest extent, we consider an example and compare the state-of-the-art caching schemes in \cite{yang2018coded,sengupta2016layered,daniel2017optimization} with our scheme. 

\begin{figure}[t]
\includegraphics[scale=0.65]{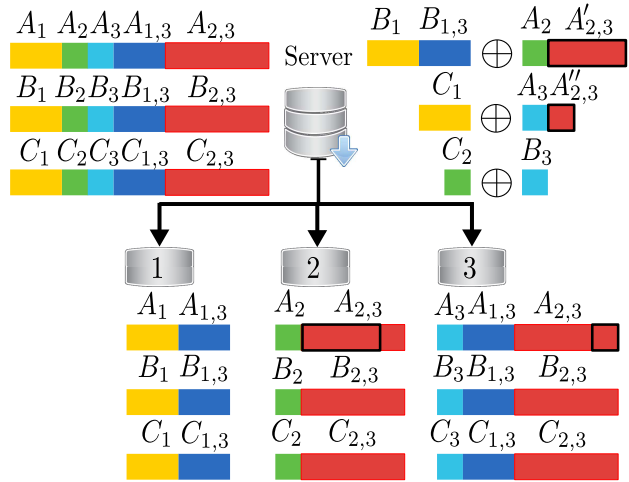}
\centering
\caption{Optimal scheme with uncoded placement for $K=N=3$ and $\bm M =[1.2,\ \! 1.5,\ \! 2.1 ]$.}\label{fig_motv}
\vspace{-0.2in}
\end{figure}
%\begin{example}\label{example_1} 
Consider a three-user system with three files, $\{A, \ \! B, \ \! C \}$, and $\bm m =[0.4, \ \! 0.5, \ \! 0.7]$. Without loss of generality, we assume that the users request files $A$, $B$, and $C$, respectively. In the placement phase, the files are divided into subfiles, which are labeled by the users exclusively storing them, e.g., subfile $A_{i,j}$ is stored at users $i$ and $j$.   
\begin{enumerate}

\item \underline{The layered scheme \cite{yang2018coded,sengupta2016layered}:} In the placement phase, the files are partitioned over three layers, we denote the files in layer $l$ by the superscript $(l)$. By optimizing the file partitioning over the layers, we get the following scheme. In layer $1$, users have equal caches with size $M_1 F$ bits and files $A^{(1)},B^{(1)},C^{(1)} $ with size $ 0.9 F$ bits, each of which is split into six disjoint subfiles, e.g., $A^{(1)}$ is divided into $A_1^{(1)}, \ \! A_2^{(1)},\ \! A_3^{(1)}, \ \! A_{1,2}^{(1)}, \ \! A_{1,3}^{(1)}, \ \! A_{2,3}^{(1)}$, where $|A_i^{(1)}|=0.2 F$, and $|A_{i,j}^{(1)}|=0.1 F$. In delivery phase, the server sends the multicast signals $B_1^{(1)} \oplus A_2^{(1)}$, $C_1^{(1)} \oplus A_3^{(1)}$, $C_2^{(1)} \oplus B_3^{(1)}$, and $C_{1,2}^{(1)} \oplus B_{1,3}^{(1)} \oplus A_{2,3}^{(1)}$. In layer $2$, we have a single user with no cache and a two-user system with file size $0.1 F$ bits and equal cache size $(M_2-M_1) F=0.1 N F$ bits. The server only needs to send a unicast signal of size $0.1 F$ bits to user $1$. In layer $3$, the $(M_3-M_2) F$ bits of the cache at user $3$ are not utilized. 
%we assign $0.9 F$ bits to layer $1$, $0.1 F$ bits to layer $2$, and layer $3$ is not utilized. 

\item \underline{The caching scheme in \cite{daniel2017optimization}:} Each file is split into six disjoint subfiles, e.g., $A$ is divided into $A_1, \ \! A_2,\ \! A_3,$ $ A_{1,2}, \ \! A_{1,3}, \ \! A_{2,3}$, where $|A_i|=0.4F/3$, $|A_{1,2}|=0.1F/3$, $|A_{1,3}|=0.7F/3$, and $|A_{2,3}|=F/3$. In delivery phase, the server sends $B_1 \oplus A_2$, $C_1 \oplus A_3$, $C_2 \oplus B_3$, and $C_{1,2} \bar \oplus B_{1,3} \bar \oplus A_{2,3}$, where $\bar \oplus$ denotes an XOR operation that allows zero padding. Note that $C_{1,2} \bar \oplus B_{1,3} \bar \oplus A_{2,3}$ can be decomposed into $C_{1,2} \oplus B_{1,3}^\prime \oplus A_{2,3}^\prime$, $B_{1,3}^{\prime \prime} \oplus A_{2,3}^{\prime \prime} $, and the unicast signal $A_{2,3}^{\prime \prime \prime} $, where $|B_{1,3}^\prime| \! = \! |A_{2,3}^\prime| \! = \! |C_{1,2}|$, $|B_{1,3}^{\prime \prime}| \! = \! |A_{2,3}^{\prime \prime}| \! = \! |B_{1,3}|\! - \! |C_{1,2}|$, and $|A_{2,3}^{\prime \prime \prime} | \! = \! |A_{2,3}| \! - \! |B_{1,3}|$.

\item \underline{Our proposed scheme:} In the placement phase, each file is split into five disjoint subfiles, e.g., $A$ is divided into $A_1, \ \! A_2,\ \! A_3, \ \! A_{1,3}, \ \! A_{2,3}$, where $|A_1|=|A_{1,3}|=0.2 F$, $|A_2|=|A_3|=0.1 F$, and $|A_{2,3}|=0.4F$. First, the server partitions $A_{2,3}$ into $A_{2,3}^\prime, \ \! A_{2,3}^{\prime \prime}$ such that $|A_{2,3}^\prime|=0.3F$ and $|A_{2,3}^{\prime \prime}|=0.1F$. Then, the server sends the multicast signals $ \big( B_1 \cup B_{1,3} \big) \oplus \big( A_2 \cup A_{2,3}^{\prime} \big)$, $C_1 \oplus \big( A_3 \cup A_{2,3}^{\prime \prime} \big)$, and $ C_2 \oplus B_3$. One can easily verify that these multicast signals enable the users to decode the requested files. The caching scheme is illustrated in Fig. \ref{fig_motv}.
\end{enumerate}

Our caching scheme achieves a delivery load equal to $ 0.7$, compared to $0.8$ by the layered scheme \cite{yang2018coded,sengupta2016layered}, and $ 0.7333$ by the scheme in \cite{daniel2017optimization}. The schemes in \cite{yang2018coded,sengupta2016layered,daniel2017optimization} need an additional unicast transmission compared with our scheme, as we have better utilization of side-information, e.g., $A_{2,3}^\prime$ is used in the multicast signal to users $\{1,2\}$. Additionally, in this example, the layered scheme does not utilize $(M_3-M_2) F$ bits of the cache at user $3$. In Theorem \ref{thm_3ue}, we show that our proposed scheme is optimal with uncoded placement. 
%\end{example}

%--------------------------------------------------------------------------------------
%--------------------------------------------------------------------------------------
\vspace{-0.1in}
\section{Cache Placement Phase }\label{sec_cach_plac}
%In this section, we introduce the set of feasible uncoded placement schemes $\mathfrak A(\bm m) $ for a given normalized cache size vector $\bm m$. For clarity of exposition, we first consider the three-user case, i.e., $K=3$, then we consider the general $K$-user case.
%\begin{figure}[t]
%\includegraphics[scale=1.2]{figures/store_subfiles.eps}
%\centering
%\caption{The partitions of file $W_{l}$ over the storage sets.}\label{fig:3Ue_subfiles}
%\end{figure}
%\subsection{The Three-User Case}\label{sec_cach_3UE_plac}
%, as shown in Fig. \ref{fig:3Ue_subfiles}.
% Furthermore, we have $|\tilde W_{l,\mc S}|= a_{\mc S} F \text{ bits}, \forall l \in [N]$, where the allocation variable $a_{\mc S} \in [0,1]$ represents the fraction of the file stored at $\mc S$. In turn, we have 
%since the sets $\mc S \subset [K]$ define a partition of each file $W_{l}$. Additionally, the users' cache sizes impose the following constraints 
%\begin{align}
%a_{\{1\}}+ a_{\{1,2\}}+a_{\{1,3\}}+a_{\{1,2,3\}}\leq m_1, \label{eqn_place_cond2} \\
%a_{\{2\}}+ a_{\{1,2\}}+a_{\{2,3\}}+a_{\{1,2,3\}}\leq m_2, \label{eqn_place_cond3} \\
%a_{\{3\}}+ a_{\{1,3\}}+a_{\{2,3\}}+a_{\{1,2,3\}}\leq m_3, \label{eqn_place_cond4}
%\end{align}
%since user $k$ allocates $m_k F$ bits of its cache memory to each file. %Consequently, the set of feasible uncoded placement schemes for a three-user system is characterized by the allocation variables $a_{\mc S}, \mc S \subset [K]$ that satisfy  (\ref{eqn_place_cond1})-(\ref{eqn_place_cond4}).
%
%\subsection{The $K$-User Case}\label{sec_cach_KUE_plac}
Each file $W_{l}$ is partitioned into $2^K$ subfiles. A subfile $\tilde W_{l,\mc S} $ is labeled by the set of users $\mc S$ exclusively storing it. The set of uncoded placement schemes for a given $\bm m$ is defined as
\begin{align}\label{eqn_feas_alloc}
 \mathfrak A(\bm m)= \bigg\lbrace \bm a \in [0,1]^{2^K} \bigg\vert \sum\limits_{\mc S \subset [K]}  a_{\mc S }=1,  \! \! \! \! \! \! \! \! \sum\limits_{\mc S \subset [K]  : \ \! k \in \mc S } \! \! \! \! \! \! \! a_{\mc S } \leq m_k, \forall k \in [K] \bigg\rbrace,
\end{align}
where $\bm a $ is the vector of allocation variables $a_{\mc S}, \ \mc S \subset [K] $ and $|\tilde W_{l,\mc S}|= a_{\mc S} F \text{ bits}, \forall l \in [N]$. For example, for $K=3$, we have
\vspace{-0.2in}
\begin{align}
a_{\phi}+a_{\{1\}}+a_{\{2\}}+a_{\{3\}}+ a_{\{1,2\}} +a_{\{1,3\}}+a_{\{2,3\}}+a_{\{1,2,3\}}= 1, \label{eqn_place_cond1} \\
a_{\{i\}}+ a_{\{i,j\}}+a_{\{i,k\}}+a_{\{i,j,k\}}\leq m_i, \ i,j,k \in \{1,2,3\}, \ i \neq j \neq k.\label{eqn_place_cond2}
\end{align}
%
%The cache placement procedure is summarized in Algorithm \ref{alg_place}, for a given $\bm a \in \mathfrak A(\bm m)$.
%\begin{algorithm}[ht]
%\begin{algorithmic}[1]
%\REQUIRE $\{ W_{1}, \dots, W_{N}\}$ and $\bm a$ 
%\ENSURE $ Z_k, k \in [K]$
%\FOR{$l \in [N]$}
%\STATE Partition file $W_{l}$ into subfiles $\tilde W_{l,\mc S}, \mc S \subset [K]$ such that $|\tilde W_{l,\mc S}|=a_{\mc S} F$. 
%\ENDFOR
%\FOR{$k \in [K]$}
%\STATE $Z_k \leftarrow \bigcup\limits_{l \in [N]} \ \bigcup\limits_{\mc S \subset [K] : \ \! k \in \mc S } \tilde W_{l,\mc S}$ 
%\ENDFOR
%\end{algorithmic}
% \caption{Cache placement procedure}\label{alg_place}
%\end{algorithm}
%\vspace{-0.1in}

%--------------------------------------------------------------------------------------
%--------------------------------------------------------------------------------------
\vspace{-.2 in}
\section{Delivery Phase}\label{sec_cach_dlv}
%In this section, we introduce the delivery scheme for systems with three users, then we generalize it to $K$-user systems.

\subsection{Multicast signals $X_{\mc T, \bm d}$}
A multicast signal $ X_{\mc T, \bm d}$ delivers a piece of the file $W_{d_j}$, $W_{d_j}^{\mc T}$, to user $j \in \mc T$. The server generates $ X_{\mc T, \bm d}$ by XORing $ W_{d_j}^{\mc T}, \ \forall j\in\mc T$, where $|W_{d_j}^{\mc T}|=v_{\mc T} F$ bits, $\forall j \in \mc T$. Each user in $\mc T\setminus \{j \}$ must be able to cancel $W_{d_j}^{\mc T}$ from $X_{\mc T, \bm d}$, in order to decode its requested piece. Consequently, $W_{d_j}^{\mc T}$ is constructed using the side-information cached by all the users in $\mc T\setminus \{j \}$ and not available at user $j$:
\vspace{-0.3in}
\begin{align}\label{eqn_KUE_mult}
X_{\mc T, \bm d}= \oplus_{j \in \mc T} \ W_{d_j}^{\mc T}= \oplus_{j \in \mc T} \bigg( \bigcup_{\mc S \in \mc B_j^{\mc T }} W_{d_j,\mc S}^{\mc T} \bigg),
\end{align}
where $W_{d_j,\mc S}^{\mc T} \subset W_{d_j}^{\mc T}$ which is stored exclusively at the users in the set $\mc S$ and 
\begin{align}\label{eqn_B_T_j}
 \mc B^{\mc T}_{j} \triangleq \Big\{ \mc S \subset [K] :  \mc T \! \setminus \! \{j\} \subset \mc S, j \not\in \mc S \Big\}, \! \ \forall j \in \mc T,
\end{align}
for example, for $K=3$ and $i,j,k \in \{1,2,3\}, \ i \neq j \neq k$, the multicast signals are defined as
\begin{align} 
 X_{\{i,j\}, \bm d}   &=   W_{d_i}^{\{i,j\}} \oplus \ W_{d_j}^{\{i,j\}} = \left(
W_{d_i,\{j\}}^{\{i,j\}}   \bigcup W_{d_i,\{j,k\}}^{\{i,j\}} \right) \oplus \left( W_{d_j,\{i\}}^{\{i,j\}}   \bigcup W_{d_j,\{i,k\}}^{\{i,j\}} \right), \label{eqn_3UE_mult1} \\
X_{\{1,2,3\}, \bm d}&= W_{d_1}^{\{1,2,3\}} \oplus \ W_{d_2}^{\{1,2,3\}} \oplus \ W_{d_3}^{\{1,2,3\}}=  W_{d_1,\{2,3\}}^{\{1,2,3\}} \oplus W_{d_2,\{1,3\}}^{\{1,2,3\}} \oplus W_{d_3,\{1,2\}}^{\{1,2,3\}}. \label{eqn_3UE_mult2}
\end{align}
where $|W_{d_i}^{\{i,j\}}| = |W_{d_j}^{\{i,j\}}|= v_{\{i,j\}} F $ bits and $|W_{d_1}^{\{1,2,3\}}| = |W_{d_2}^{\{1,2,3\}}|= |W_{d_3}^{\{1,2,3\}}|= v_{\{1,2,3\}} F $. $| W_{d_j,\mc S}^{\mc T}|= u^{\mc T}_{\mc S} F$ bits, i.e., the assignment variable $ u^{\mc T}_{\mc S} \in [0,a_{\mc S}]$ represents the fraction of $\tilde W_{d_j,\mc S}$ involved in the multicast signal $X_{\mc T, \bm d}$. Note that one subfile can contribute to multiple multicast transmissions, for example in a three-user system $\tilde W_{d_k,\mc \{i,j\}}$ is used in $X_{\{i,k\}, \bm d}$, $X_{\{j,k\}, \bm d}$, $X_{\{i,j,k\}, \bm d}$. Therefore, in order to guarantee that no redundant bits are transmitted, each subfile $\tilde W_{d_k,\mc S}$ is partitioned into disjoint pieces, e.g., for $K=3$, we have
\vspace{-0.1in}
\begin{align}
\tilde W_{d_k,\mc \{i,j\}} &=  W_{d_k,\mc \{i,j\}}^{\{i,k\}} \bigcup W_{d_k,\mc \{i,j\}}^{\{j,k\}} \bigcup W_{d_k,\mc \{i,j\}}^{\{i,j,k\}} \bigcup  W_{d_k,\mc \{i,j\}}^{\phi}, \label{eqn_3UE_prtn1} 
\end{align}
where $W_{d_k,\mc S}^{\phi}$ denotes the remaining piece which is not involved in any transmission. 
\begin{remark} By contrast with \cite{maddah2014fundamental,yang2018coded,sengupta2016layered,daniel2017optimization}, where multicast signals of size $t+1$ utilize only the side-information stored exclusively at $t$ users, i.e., $ X_{\mc T,\bm d}=  \oplus_{k \in \mc T} W^{\mc T}_{d_k,\mc T \setminus \{k\}}$, the structure of the multicast signal in (\ref{eqn_KUE_mult}) represents all feasible utilizations of the side-information. This flexibility is instrumental in achieving the delivery load memory trade-off with uncoded placement $R^*_{\mathfrak{A}}$. 
\end{remark}
\vspace{-0.2in}
\subsection{Unicast signals $X_{\{ i \}}$}
A unicast signal $X_{\{ i \}}$ delivers the fraction of the requested file which is not stored at user $i$ and will not be delivered by the multicast transmissions. For example, for $K=3$, we have 
\begin{align}
X_{\{i\}, \bm d}&= W_{d_i} \! \setminus  \Big(   \bigcup_{\mc S : i \in \mc S} \! \! \! \tilde W_{d_i,\mc S} \ \! \bigcup   W_{d_i}^{\{i,j\}}  \bigcup W_{d_i}^{\{i,k\}} \bigcup W_{d_i}^{\{i,j,k\}} \Big), \ \! i,j,k \in \{1,2,3\}, \ i \neq j \neq k\label{eqn_3UE_uni} 
\end{align}
where $\bigcup_{\mc S : i \in \mc S} \tilde W_{d_i,\mc S} $ is stored at user $i$ and $ W_{d_i}^{\mc T} $ is delivered to user $i$ via $X_{\mc T, \bm d}$.

% @@@@@@@@@@@@@@@@@@@@@@@@@@@@@@@@@@@@@@@@@@@@@@@@@@@@@@@@@@@@@@@@@@@@@@@@@@@@@@@@@@@
\vspace{-0.2in}
\subsection{Delivery phase constraints} Recall that $v_{\mathcal{T}} \in [0, \ \! 1]$ and  $ u^{\mc T}_{\mc S} \in [0,a_{\mc S}]$ represent $|X_{\mc T, \bm d}|/F$, and $| W_{d_j,\mc S}^{\mc T}|/F$, respectively. Our delivery scheme can be represented by constraints on $v_{\mathcal{T}}$ and $u^{\mc T}_{\mc S}$ as follows.
% @@@@@@@@@@@@@@@@@@@@@@@@@@@@@@@@@@@@@@@@@@@@@@@@@@@@@@@@@@@@@@@@@@@@@@@@@@@@@@@@@@@
First, the structure of the multicast signals in (\ref{eqn_3UE_mult1}), (\ref{eqn_3UE_mult2}) imposes 
\begin{align}\label{eqn_KUE_struct}
\sum_{\mc S \in \mc B_j^{\mc T } } u^{\mc T}_{\mc S} = v_{\mc T}, \ \forall \ \! \mc T \subsetneq_{\phi} [K], \ \forall \! \  j \in \mc T.
\end{align}
For example, for $K=3$, we have
\begin{align}
v_{\{i,j\}} &= u^{\{i,j\}}_{\{j\}}+u^{\{i,j\}}_{\{j,k\}} = u^{\{i,j\}}_{\{i\}}+u^{\{i,j\}}_{\{i,k\}}, \  
v_{\{1,2,3\}} = u^{\{1,2,3\}}_{\{2,3\}}= u^{\{1,2,3\}}_{\{1,3\}}= u^{\{1,2,3\}}_{\{1,2\}}. \label{eqn_3UE_delv4}
\end{align}
% $X_{\mc T, \bm d}$
% @@@@@@@@@@@@@@@@@@@@@@@@@@@@@@@@@@@@@@@@@@@@@@@@@@@@@@@@@@@@@@@@@@@@@@@@@@@@@@@@@@@ 
In order to prevent transmitting redundant bits from the subfile $\tilde W_{d_j,\mc S}$ to user $j$, we need 
\begin{align}\label{eqn_KUE_redund}
\sum\limits_{\mc T \subsetneq_{\phi} [K] : \ \! j \in \mc T, \mc T \cap \mc S \neq \phi,  \mc T \setminus \{j\} \subset \mc S } \! \! \! \! \! \! \! \! \! \! \!  u^{\mc T }_{\mc S} \leq a_{\mc S}, \! \ \forall \! \  j \not\in \mc S, \ \forall \! \  \mc S \in \left\lbrace \tilde{\mc S} \subset [K]: \ 2 \leq |\tilde{\mc S}| \leq K-1 \right\rbrace,
\end{align}
 
where the condition $\mc T \setminus \{j\} \subset \mc S $ follows from (\ref{eqn_B_T_j}). %In turn, $u^{\mc T }_{\mc S} $ is defined only for $|\mc T| \leq |\mc S|+1  $. 

For example, for $K=3$, (\ref{eqn_KUE_redund}) implies
\begin{align}
u^{\{i,k\}}_{\{i,j\}}+u^{\{j,k\}}_{\{i,j\}}+u^{\{i,j,k\}}_{\{i,j\}} &\leq a_{\{i,j\}}. \label{eqn_3UE_delv5} 
\end{align} 
% @@@@@@@@@@@@@@@@@@@@@@@@@@@@@@@@@@@@@@@@@@@@@@@@@@@@@@@@@@@@@@@@@@@@@@@@@@@@@@@@@@@
Finally, the delivery signals sent by the server must complete all the requested files:
\begin{align}\label{eqn_KUE_delv}
\sum\limits_{\mc T \subsetneq_{\phi} [K] : \ \! k \in \mc T } v_{\mc T }  \geq 1-\! \! \! \! \sum_{\mc S \subset [K] : \ \! k \in \mc S} \! \! \! \! a_{\mc S}, \forall \! \  k \in [K],
\end{align}
for example, for $K=3$, the delivery completion constraint for user $i$ is given by
\begin{align}
v_{\{i\}}+ v_{\{i,j\}}+v_{\{i,k\}}+v_{\{i,j,k\}}\geq 1-\! (a_{\{i\}}+ a_{\{i,j\}}+a_{\{i,k\}}+a_{\{i,j,k\}}). \label{eqn_3UE_delv14}
\end{align}
% \! \! \! \sum_{\mc S \subset [K]  : \ \! 1 \in \mc S} \! \! \! \! a_{\mc S}
%In summary, for a three-user system, the set of feasible linear delivery schemes is characterized by (\ref{eqn_3UE_delv1})-(\ref{eqn_3UE_delv16}) and $0 \leq u^{\mc T}_{\mc S} \leq a_{\mc S}$.
%
%\vspace{-.07 in}
%\subsection{The $K$-User Case}\label{sec_cach_KUE_dlv}

Therefore, for given $\bm a$, the set of feasible delivery schemes, $\mathfrak D (\bm a)$, is defined as
\begin{align}\label{eqn_feas_delv}
\mathfrak D (\bm a) \! =  \Bigg\lbrace (\bm v, \bm u) &\bigg\vert \sum\limits_{\mc T \subsetneq_{\phi} [K] : \ \! k \in \mc T } \! \! \! \! \! \! \! \! \! \! v_{\mc T } \geq 1-\! \! \! \! \! \! \! \! \sum_{\mc S \subset [K] : \ \! k \in \mc S} \! \! \! \! \! \! \! \! a_{\mc S}, \ \! \forall k \in [K], \sum_{\mc S \in \mc B_j^{\mc T } } \! \! u^{\mc T}_{\mc S} = v_{\mc T}, \ \forall \mc T \subsetneq_{\phi} [K], \ \! \forall j \in \mc T,   \nonumber \\ &\sum\limits_{\mc T \subsetneq_{\phi} [K] : \ \! j \in \mc T, \mc T \cap \mc S \neq \phi, \mc T \setminus \{j\} \subset \mc S } \! \! \! \! \! \! \! \! \! u^{\mc T }_{\mc S} \leq a_{\mc S}, \ \forall j \not\in \mc S, \forall \mc S \in \left\lbrace \tilde{\mc S} \subset [K]: \ 2 \leq |\tilde{\mc S}| \leq K-1 \right\rbrace, \nonumber \\ &0 \leq u_{\mc S}^{\mc T} \leq a_{\mc S}, \! \ \forall \mc T \subsetneq_{\phi} [K], \ \forall \mc S \in \bigcup\limits_{j \in \mc T} \mc B^{\mc T}_{j} \Bigg\rbrace,
\end{align}
where the transmission and assignment variables are represented by $\bm v$ and $\bm u$ respectively. %Finally, the delivery procedure is summarized in Algorithm \ref{alg_delv}.
%
%\begin{algorithm}[t]
%\begin{algorithmic}[1]
%\REQUIRE $\bm d, \bm a, \bm u, \bm v$, and $\tilde W_{l,\mc S}, \mc S \subset [K], l \in [N]$ 
%\ENSURE $ X_{\mc T, \bm d}, \mc T \subsetneq_{\phi} [K]$
%
%\COMMENT{Partitioning}
%\FOR{$\{ \mc S \subset [K] : 1 \leq |\mc S|\leq K-1 \}$}
%\FOR{$ \{j \in [K]: j \not\in \mc S\}$}
%\STATE Partition $\tilde W_{d_j,\mc S}$ into $ W_{d_j,\mc S}^{\mc T},$ $ \big\lbrace \mc T \subsetneq_{\phi} [K] : j \in \mc T, \mc T \cap \mc S \neq \phi \big\rbrace$ and $ W_{d_j,\mc S}^{\phi}$, such that $| W_{d_j,\mc S}^{\mc T}|=u^{\mc T}_{\mc S} $ and $W_{d_j,\mc S}^{\phi}$ is the remaining segment.  
%\ENDFOR
%\ENDFOR
%
%\COMMENT{Delivery scheme}
%\FOR{$\mc T  \subsetneq_{\phi} [K] $} 
%\IF{$\mc T=\{j\} $} 
%\STATE $X_{\{j\}, \bm d} \leftarrow W_{d_j} \setminus \bigg( \Big( \bigcup\limits_{\mc S : j \in \mc S} \tilde W_{d_j,\mc S }\Big) \bigcup \Big( \bigcup\limits_{\mc T^{\prime},\mc S} W_{d_j,\mc S}^{\mc T^{\prime}} \Big) \bigg)$ \COMMENT{Unicast transmissions}
%\ELSE
%\STATE $X_{\mc T, \bm d} \leftarrow \oplus_{j \in \mc T} \ \bigcup\limits_{\mc S \in \mc B_j^{\mc T } } W_{d_j,\mc S}^{\mc T}$  \COMMENT{Multicast transmissions}
%\ENDIF
%\ENDFOR
%\end{algorithmic}
% \caption{Delivery procedure}\label{alg_delv}
%\end{algorithm}

\vspace{-0.15in}
\subsection{Discussion}
\vspace{-0.05in}
%Next, we explain how the linear constraints in (\ref{eqn_feas_delv}) guarantee the delivery of the requested files. In particular, successful delivery is guaranteed by the following conditions:
%\begin{enumerate*}
%\item User $j \in \mc T$ can retrieve $W_{d_j}^{\mc T} $ from the signal $X_{\mc T, \bm d}$, which follows from the definition of the structural constraints in (\ref{eqn_KUE_struct}).
%\item The requested file can be reconstructed from the pieces decoded at the user, which follows from the redundancy constraints in (\ref{eqn_KUE_redund}) and delivery completion constraints in (\ref{eqn_KUE_delv}). More specifically, the delivery completion constraints ensure that the number of decoded bits are sufficient for decoding the file, and the redundancy constraints prevent the server from transmitting redundant bits.
%\end{enumerate*}
%Note that the amount of side-information at the users limits the size of multicast transmissions, which is necessary in order to guarantee decodability at the users. The next proposition shows that the structural and redundancy constraints in (\ref{eqn_KUE_struct}) and (\ref{eqn_KUE_redund}), ensure that these side-information constraints are satisfied.
%\vspace{-0.1in}
The linear constraints in (\ref{eqn_feas_delv}) guarantee the delivery of the requested files. Successful delivery is guaranteed by 
\begin{enumerate*}
\item By (\ref{eqn_KUE_struct}), user $j \in \mc T$ can retrieve $W_{d_j}^{\mc T} $ from the signal $X_{\mc T, \bm d}$.
\item By (\ref{eqn_KUE_redund}) and (\ref{eqn_KUE_delv}), $W_{d_j} $ can be reconstructed from the pieces decoded at user $j$. The delivery completion constraints ensure that the number of decoded bits are sufficient for decoding the file, and the redundancy constraints prevent the server from transmitting redundant bits.
\end{enumerate*} Formally, we have:
%Note that the amount of side-information at the users limits the size of multicast transmissions, which is necessary in order to guarantee decodability at the users. The next proposition shows that the structural and redundancy constraints in (\ref{eqn_KUE_struct}) and (\ref{eqn_KUE_redund}), ensure that these side-information constraints are satisfied.
\begin{proposition}\label{prop_side_info}
For $\mc S^\prime \! \subset [K]$ such that $ 1 \leq |\mc S^\prime| \leq K \! - \! 2$, and some user $j \not\in \mc S^\prime$, the size of the multicast transmissions $X_{\mc T, \bm d}, $ where $ \{j\} \cup \mc S^\prime \subset \mc T,$ is limited by the amount of side-information stored at the users in $\mc S^\prime $ and not available at user $ j$, i.e., 
\begin{align}
\sum\limits_{\mc T \subsetneq_{\phi} [K] : \ \! \{j\} \cup \mc S^{\prime} \subset \mc T}  \! \! \! \! \!  \! \! \! \! v_{\mc T}  \leq  \! \! \!  \sum_{\mc S \subset [K] : \ \mc S^{\prime}  \subset \mc S, j \not\in \mc S} \! \! \! \! \! \! \! \! \! \! \! a_{\mathcal S}, 
\end{align}
which is guaranteed by (\ref{eqn_KUE_struct}) and (\ref{eqn_KUE_redund}).
\end{proposition} \vspace{-0.1in}
The proof of Proposition \ref{prop_side_info} is provided in Appendix \ref{app_side_info}.
%\begin{proof}
%See Appendix \ref{app_side_info}.
%\end{proof}

%--------------------------------------------------------------------------------------
%--------------------------------------------------------------------------------------
\vspace{-0.1in}
\section{Formulation and Results}\label{sec_cach_opt}
In this section, we first show that the optimal uncoded placement and linear delivery schemes can be obtained by solving a linear program. Next, we present a lower bound on the delivery load with uncoded placement. Based on this bound, we show that linear delivery is optimal with uncoded placement for three cases; namely, $\sum_{k=1}^{K} m_k \leq 1$, $\sum_{k=1}^{K} m_k \geq K\!-\!1$, and the three-user case. That is, for these cases we explicitly characterize the delivery load memory trade-off with uncoded placement $R^*_{\mathfrak A} (\bm m)$.  
\subsection{Caching Scheme Optimization}
In Sections \ref{sec_cach_plac} and \ref{sec_cach_dlv}, we have demonstrated that an uncoded placement scheme in $\mathfrak A$ is completely characterized by the allocation vector $\bm a$, which represents the fraction of files stored exclusively at each subset of users $\mc S \subset [K]$. Additionally, the assignment and transmission vectors $(\bm u, \bm v)$ completely characterize a delivery scheme in $\mathfrak D$, where $\bm v$ represents the size of the transmitted signals, and $\bm u$ determines the structure of the transmitted signals. For a given normalized memory vector $\bm m$, the following optimization problem characterizes the minimum worst-case delivery load $R^*_{\mathfrak{A},\mathfrak{D}}(\bm m) $ and the optimal caching scheme in $\mathfrak A, \mathfrak D$, i.e., the optimal values for $\bm a$, $\bm v$, and $\bm u$.
\begin{subequations} \label{eqn_opt1}
	\begin{align}
	\textit{\textbf{O\arabic{opt_ct}}:}  \qquad  & \min_{\bm a,\bm u ,\bm v}  
	& & \sum_{\mc T \subsetneq_{\phi} [K]} v_{\mc T} \\
	& \text{s.t.}
	& & \bm a \in \mathfrak{A}(\bm m), \text{ and } (\bm u, \bm v) \in \mathfrak D(\bm a),
	\end{align}
\end{subequations}
\stepcounter{opt_ct}%
where $\mathfrak{A}(\bm m) $ and $\mathfrak D(\bm a) $ are defined in (\ref{eqn_feas_alloc}) and (\ref{eqn_feas_delv}), respectively. 
\begin{remark} 
For equal cache sizes, $R^*_{\mathfrak{A},\mathfrak{D}}(\bm m)$ is equal to the worst-case delivery load of \cite{maddah2014fundamental}, which was shown to be optimal for uncoded placement in \cite{wan2016optimality} for $N \geq K$. For $N<K$, the optimal scheme for uncoded placement was proposed in \cite{yu2016exact}. The solution of (\ref{eqn_opt1}) is equivalent to the memory sharing solution proposed in \cite{maddah2014fundamental}.
\end{remark}
% , optimal for uncoded placement \cite{wan2016optimality,yu2016exact}.
\begin{remark}\label{remark_padd}
In Section \ref{sec_cach_dlv}, $X_{\mc T, \bm d}$ is formed by XORing pieces of equal size. A delivery scheme with $\tilde X_{\mc T, \bm d}=  \bar \oplus_{j \in \mc T} \  W_{d_j}^{\mc T} $, where $\bar \oplus$ denotes an XOR operation that allows zero padding so that the pieces are of equal length, is equivalent to a delivery scheme in $\mathfrak{D} $ and both yield the same delivery load. For example, $\tilde X_{\{1,2\}, \bm d}=   W_{d_1, \{2\}}^{\{1,2\}}  \bar \oplus \ W_{d_2, \{1\}}^{\{1,2\}} $, with $ u_{\{2\}}^{\{1,2\}}>u_{\{1\}}^{\{1,2\}} $, is equivalent to a multicast signal $X_{\{1,2\}, \bm d}$ and a unicast signal $X_{\{2\}, \bm d}$, with sizes $ u_{\{1\}}^{\{1,2\}} F$ bits, and $ \big( u_{\{2\}}^{\{1,2\}}-u_{\{1\}}^{\{1,2\}} \big) F$ bits, respectively. 
\end{remark}
\vspace{-0.2in}

\begin{remark}\label{remark_subpack}
In this work, we assume the file size to be large enough, such that the cache placement and delivery schemes can be tailored to the unequal cache sizes by optimizing over continuous variables. More specifically, for $F$ large enough, $ u_{\mc S}^{\mc T} F $ can be used instead of $\lceil u_{\mc S}^{\mc T} F \rceil$ for $u_{\mc S}^{\mc T} \in[0,1]$. The required subpacketization level is the least-common-denominator of the assignment variables $ u_{\mc S}^{\mc T} $. That is, the minimum packet size is equal to the greatest-common-divisor (gcd) of all $u_{\mc S}^{\mc T}F$, assuming $u_{\mc S}^{\mc T}F$ are integers.
\end{remark}

%Now, we consider a system where the users have equal cache sizes in order to show that the solution obtained from (\ref{eqn_opt1}) is the same as the memory sharing solution proposed in \cite{maddah2014fundamental}.
%\begin{example}
%Consider a caching system with $\bm m=[0.5, \ 0.5, \ 0.5]$, i.e., $m_k \in [\frac{1}{3},\frac{2}{3}]$. The caching scheme obtained from (\ref{eqn_opt1}), is characterized as follows.\newline
%\underline{Placement phase}: We have $a_{\{1\}} \! = \! a_{\{2\}} \! = \! a_{\{3\}} \! = \! 1/6$ and $a_{\{1,2\}} \! = \! a_{\{1,3\}} \! = \! a_{\{2,3\}} \! = \! 1/6$.% , as shown in Fig. .
%\newline 
%\underline{Delivery phase}: We have the following transmissions
%\begin{align*}
%X_{\{1,2\}, \bm d}&= W_{d_1,\{2\}}^{\{1,2\}} \oplus \ W_{d_2,\{1\}}^{\{1,2\}}, \ 
%X_{\{1,3\}, \bm d}= W_{d_1,\{3\}}^{\{1,3\}} \oplus \ W_{d_3,\{1\}}^{\{1,3\}}, \ 
%X_{\{2,3\}, \bm d}= W_{d_2,\{3\}}^{\{2,3\}} \oplus \ W_{d_3,\{2\}}^{\{2,3\}}, \\
%X_{\{1,2,3\}, \bm d}&= W_{d_1,\{2,3\}}^{\{1,2,3\}} \oplus \ W_{d_2,\{1,3\}}^{\{1,2,3\}} \oplus \ W_{d_3,\{1,2\}}^{\{1,2,3\}},
%\end{align*}
%where $v_{\{1,2\}} \! = \! v_{\{1,3\}} \! = \! v_{\{2,3\}} \! = \! 1/6$, and $v_{\{1,2,3\}} \! = \! 1/6$. The minimum worst-case delivery load $R^*_{\mathfrak{A}}(\bm m)=R^*_{\mathfrak{A},\mathfrak{D}}(\bm m)=2/3$, as shown in \cite{maddah2014fundamental,yu2016exact,wan2016optimality}. %, and the breakdown of the requested files is illustrated in Fig. . 
%\end{example}

%--------------------------------------------------------------------------------------
\vspace{-0.2in}
\subsection{Lower Bounds}
Next, using the approach in \cite{wan2016optimality,yu2016exact}, we show that $R^*_{\mathfrak{A}} (\bm m)$ is lower bounded by the linear program in (\ref{eqn_bound_genie}).
\begin{theorem}\label{thm_bound_genie} (Uncoded placement bound) Given $K$, $N \geq K$, and $\bm m$, the minimum worst-case delivery load with uncoded placement $R^*_{\mathfrak{A}}(\bm m)$ is lower bounded by 
\begin{subequations} \label{eqn_bound_genie}
	\begin{align}
\textit{\textbf{O\arabic{opt_ct}}:}  \qquad  & \max_{\lambda_{0} \in \mathbb{R},\lambda_{k} \geq 0,\alpha_{\bm q} \geq 0}  
	& & -\lambda_{0} - \sum_{k=1}^{K} m_k \lambda_{k} \\
	& \ \ \ \ \ \ \text{s.t.}
	& &  \! \! \! \! \lambda_0+ \sum_{k \in \mc S} \lambda_k + \gamma_{\mc S} \geq 0, \ \! \forall \mc S \subset [K], \text{ and } \sum\limits_{\bm q \in  \mc P_{[K]} } \alpha_{\bm q} =1,
	\end{align}
\end{subequations}
\stepcounter{opt_ct}
where
\vspace{-0.3in} 
\begin{align}\label{eqn_gamma_S}
\gamma_{\mc S} \triangleq \begin{cases} K, \text{ for } \mc S =\phi,\\
0, \text{ for } \mc S =[K], \\
\sum\limits_{j=1}^{K-|\mc S|}  \sum\limits_{\substack{\bm q \in  \mc P_{[K]}   : \ \!  q_{j+1}  \in \mc S, \\ \{ q_1,\dots,q_{j} \} \cap \mc S=\phi } }  j \ \! \alpha_{\bm q}, \text{ otherwise}.
\end{cases}
\end{align}
and $\mc P_{[K]} $ is the set of all permutations of the vector $[1,2,\dots,K]$.
\end{theorem}
The proof of Theorem \ref{thm_bound_genie} is provided in Appendix \ref{app_thm_genie}. 
We compare the achievable delivery load $R^*_{\mathfrak{A},\mathfrak{D}}(\bm m) $ with the following lower bounds on the worst-case delivery load $R^* $. From \cite{amiri2017decentralized}, $R^* (\bm m)$ is lower bounded by 
\vspace{-0.1in}
\begin{align}\label{eqn_bound_amiri}
\max_{s \in [K], \ \! l \in \left[ \lceil \frac{N}{s} \rceil \right] } \Bigg\lbrace \! \frac{N-(N \! - \! K l)^{+}}{l}   -   \frac{s N \sum\limits_{i=1}^{s  +  \gamma} m_i  +\gamma (N \! - \! ls)^{+} }{l(s\! + \! \gamma)} \Bigg\rbrace,
\end{align}
where $\gamma  \! \triangleq   \min \! \left\lbrace\! \left(\lceil\frac{N}{l}\rceil-s\right)^{+}\!\!,K-s \right\rbrace$ and $m_1 \leq \dots \leq m_K$. The following proposition is a straightforward generalization of the lower bounds in \cite{wang2017improved} for systems with distinct cache sizes.
\begin{proposition} Given $K$, $N $, $\bm m$, and $m_1 \leq \dots \leq m_K$, we have 
\begin{align}\label{eqn_bound_wang}
R^*(\bm m) \geq \max \left\lbrace  \max_{s \in [\min\{K,N\}]} \left\lbrace s- \sum_{k=1}^{s}  \dfrac{N \sum_{i=1}^{k} m_i}{N-k+1} \right\rbrace, \! \ \max_{s \in [\min\{K,N\}]} \left\lbrace s \Big(1- \sum_{i=1}^{s} m_i \Big) \right\rbrace \right\rbrace.
\end{align}
\vspace{-0.2in}
\end{proposition}

\vspace{-0.2in}
%--------------------------------------------------------------------------------------
\subsection{Special Cases}
Next, we consider three special cases, for which we explicitly characterize $R^*_{\mathfrak{A}}(\bm m)$ and show that the solution of (\ref{eqn_opt1}) is the optimal caching scheme with uncoded placement. First, we consider the small cache regime, $\sum_{i=1}^{K} m_i  \leq 1$,
% The following theorem characterizes $R^*_{\mathfrak{A}}(\bm m)$ for the case where the sum of the cache sizes is less than or equal to the library size, i.e., $\sum_{i=1}^{K} m_i  \leq 1$.
\begin{theorem}\label{thm_spchcase}
The minimum worst-case delivery load with uncoded placement is given by 
\begin{align}\label{eqn_thm_spchcase}
R^*_{\mathfrak{A}}(\bm m)= K - \sum_{j=1}^{K} (K-j+1) m_j,
\end{align}
for $m_1 \leq \dots \leq m_K$, $\sum_{i=1}^{K} m_i  \leq 1$, and $N \geq K$.
%which is achieved by the placement scheme described by $a^*_{\{j\}}=m_j$, and the delivery scheme defined by $v^*_{\{j\}}=1-\sum\limits_{i=1}^{j-1} m_i - (K-j+1) m_j$, and $v^*_{\{i,j\}}=u^{\! \! * {\{i,j\}}}_{\{i\}}= u^{\! \! *{\{i,j\}}}_{\{j\}}=\min \{ a^*_{\{i\}},a^*_{\{j\}} \}$. 
\end{theorem}
\begin{proof} \textbf{Achievability:} In the placement phase, each file is split into $K+1$ subfiles such that $a_{\{j\}}=m_j$ and $a_{\phi}=1-\sum_{k=1}^{K} m_k $. In the delivery phase, we have $v_{\{j\}}=1-\sum_{i=1}^{j-1} m_i - (K-j+1) m_j$ and $v_{\{i,j\}}=u^{{\{i,j\}}}_{\{i\}}= u^{{\{i,j\}}}_{\{j\}}=\min \{ a_{\{i\}},a_{\{j\}} \}$. In turn, $R_{\mathfrak{A},\mathfrak{D}}(\bm m) = K - \sum_{j=1}^{K} (K-j+1) m_j $ is achievable.
\textbf{Converse:} By substituting $\alpha_{[1,2,\dots,K]}=1$ in Theorem \ref{thm_bound_genie}, we get 
\begin{subequations} \label{eqn_thm2_conv1}
	\begin{align}
	& \max_{\lambda_{k} \geq 0,\lambda_{0}}  
	& &  -\lambda_0-\sum_{k=1}^{K} m_k \ \! \lambda_k \\
	& \text{s.t.}
	& & \lambda_0+ \sum_{k \in \mc S} \lambda_k + \gamma_{\mc S} \geq 0, \ \! \forall \mc S \subset [K],
	\end{align}
\end{subequations}
where $\gamma_{\mc S} = j\!-\!1$ if $\{j\} \in \mc S$ and $[j-1] \cap \mc S = \phi$ for $j \in [K]$. $\lambda_0=-K, \lambda_j=K-j+1$ is a feasible solution to (\ref{eqn_thm2_conv1}), since $\lambda_0+\lambda_j+(j-1)=0$. Therefore, $R^*_{\mathfrak{A}}(\bm m) \geq K - \sum_{j=1}^{K} (K-j+1) m_j$.
\end{proof}
Next theorem characterizes $R^*_{\mathfrak{A}}(\bm m)$ in the large total memory regime where $\sum_{i=1}^{K} m_i  \geq K\!-\!1$.% In particular, we show that our caching scheme is optimal with general placement.
\begin{theorem}\label{thm_large_mem}
The minimum worst-case delivery load with uncoded placement is given by 
\begin{align}\label{eqn_thm_large_mem}
R^*_{\mathfrak{A}}(\bm m)=R^*(\bm m) = 1-m_1,
\end{align}
for $m_1 \leq \dots \leq m_K$, $\sum_{i=1}^{K} m_i  \geq K\!-\!1$, and $N \geq K$.  
\end{theorem}
\begin{proof}
\textbf{Achievability:} In the placement phase, $W_{j}$ is partitioned into subfiles $  \tilde W_{j,[K]\setminus \{i\}}, i \in [K]$ and $\tilde W_{j,[K]} $, such that $a_{[K]}=\sum\limits_{i=1}^{K}m_i -(K\!-\!1)$ and $a_{[K]\setminus \{i\}}=1-m_i, \ \! i \in [K].$
%\begin{align}
%&a_{[K]}=\sum_{i=1}^{K}m_i -(K\!-\!1), \\
%&a_{[K]\setminus \{i\}}=1-m_i, \ \! i \in [K]. 
%\end{align}
In the delivery phase, we have the following cases.
\begin{itemize}[leftmargin=*]
\item For $(K\!-\!2)m_1 \geq \sum_{i=2}^{K}m_i-1$, we have the following multicast transmissions 
\vspace{-0.2in}
\begin{align}
X_{[K]\setminus \{i\}, \bm d}&= \oplus_{k \in [K]\setminus \{i\}} W^{[K]\setminus \{i\}}_{d_k,[K]\setminus \{k\}}, \ i \in \{2,\dots,K\}, \\
X_{[K], \bm d}&=\oplus_{k \in [K]} W^{ [K]}_{d_k,[K]\setminus \{k\}},
\end{align}
with $v_{[K]\setminus \{i\}}= m_i-m_1, \ i \in \{2,\dots,K\}, $ and $ v_{[K]}=1+(K-2)m_1- \sum\limits_{k=2}^{K}m_k$. %and $a_{[K]\setminus \{i\}}=\sum\limits_{\mc T \subset [K] \ \! : \ \! |\mc T|\geq K-1, \ i \in \mc T} u^{\mc T}_{[K]\setminus \{i\}}, \ \! i \in [K]. $
\vspace{+0.05in}
\item For $(K \! - \! l- \! 1)m_l < \sum_{i=l+1}^{K} m_i \! - \! 1$ and $(K \! - \! l- \! 2)m_{l+1} \geq \sum_{i=l+2}^{K} m_i \! - \! 1$, where $ l \in [K \!- \!2]$, we have the following transmissions
\vspace{-0.2in}
\begin{align}
X_{[i], \bm d}&= \oplus_{k \in [i]} W^{[i]}_{d_k,[K]\setminus \{k\}}, \ i \in [l], \\
X_{[K]\setminus \{i\}, \bm d}&=\oplus_{k \in [K]\setminus \{i\}} W^{[K]\setminus \{i\}}_{d_k,[K]\setminus \{k\}}, \ \! i \in \{l\!+\!1,\dots,K\}.
\end{align}
with $v_{[i]} = m_{i+1}-m_i, \ \! i \in [l\!-\!1], $ $v_{[l]} = \dfrac{1}{K\!-\!l\!-\!1} \bigg( \sum_{j=l+1}^{K} m_j -1 -(K\!-\!l\!-\!1)m_l \bigg), $ and $v_{[K]\setminus \{i\}} = \dfrac{1}{K\!-\!l\!-\!1} \bigg( (K\!-\!l\!-\!1)m_i +1 -\sum_{j=l+1}^{K} m_j  \bigg), \ \! i \in \{l\!+\!1,\dots,K\}. $
\end{itemize}
In both cases, the size of the assignment variables satisfies $u_{[K]\setminus\{k\}}^{\mc T}=v_{\mc T}, \forall k \in \mc T $.

\textbf{Converse:} By substituting $s=1$ in (\ref{eqn_bound_wang}), we get $R^*(\bm m) \geq 1-m_1 $.
\end{proof}
In the next theorem, we characterize $R^*_{\mathfrak{A}}(\bm m)$ for $K=3$.
\begin{theorem}\label{thm_3ue}
For $K=3$, $N \geq 3$, and $m_1 \leq m_2 \leq m_3$, the minimum worst-case delivery load with uncoded placement 
\begin{align}\label{eqn_thm_3ue}
\! \! \! \! R^*_{\mathfrak{A}}(\bm m) \! =  \! \max \left\lbrace \! 3 \! - \! \sum_{j=1}^{3} (4-j) m_j, \ \! \frac{5}{3} \! - \! \sum_{j=1}^{3} \frac{(4-j) m_j}{3}, \ \! 2 \! - \! 2 m_1 - \! m_2, \ \! 1 \! - \! m_1 \! \right\rbrace \! .
\end{align}
In particular, we have the following regions
\begin{enumerate}
\item For $\sum\limits_{j=1}^{3} m_j \! \leq \! 1$, $R^*_{\mathfrak{A}}(\bm m) \! = \! 3 \! - \! \sum\limits_{j=1}^{3} (4-j) m_j. $
\item For $1 \! \leq \! \sum\limits_{j=1}^{3} m_j \! \leq \! 2$, we have three cases
\begin{itemize}
\item If $m_3 \! < \! m_2 \! + \! 3 m_1 \! - \! 1 $, and $2 m_2 \! + \! m_3 \! < \! 2$, then $R^*_{\mathfrak{A}}(\bm m) \! = \! \dfrac{5}{3} \! - \! \sum\limits_{j=1}^{3} \dfrac{(4-j) m_j}{3}. $
\item If $m_3 \! \geq \! m_2 \! + \! 3 m_1 \! - \! 1 $, and $m_1 \! + \! m_2 \! < \!1$, then  $R^*_{\mathfrak{A}}(\bm m) \! = 2 \! - \! 2 m_1 - \! m_2 $.
\item If $2 m_2 \! + \! m_3 \! \geq \! 2$, and $m_1 \! + \! m_2 \! \geq \! 1$, then $R^*_{\mathfrak{A}}(\bm m) \! = \! 1 \! - \! m_1$.
\end{itemize}
\item For $\sum\limits_{j=1}^{3} m_j \geq 2$, $R^*_{\mathfrak{A}}(\bm m) \! = \! 1 \! - \! m_1$.
\end{enumerate}
\vspace{-0.34in}
\end{theorem}
Proof of Theorem \ref{thm_3ue} is provided in Appendix \ref{app_thm_3ue}.
%\begin{proof}
%\textbf{Achievability:} In Appendix \ref{app_thm_3ue}, we show that for each of the four regions in (\ref{eqn_thm_3ue}) there exists a feasible solution to (\ref{eqn_opt1}) that achieves the corresponding delivery load. \textbf{Converse:} Based on Theorem \ref{thm_bound_genie}, there exits $K!$ bounding cuts for $R^*_{\mathfrak A}$. Each of the four regions in (\ref{eqn_thm_3ue}) corresponds to a convex combination of these cuts, see Appendix \ref{app_thm_3ue}.
%\end{proof}  
\begin{remark} By substituting $m_3=1$ in Theorem \ref{thm_3ue}, we obtain the two-user delivery load memory trade-off with uncoded placement, given as $R^*_{\mathfrak{A}}(\bm m) \! =  \! \max \left\lbrace 2 \! - \! 2 m_1 - \! m_2, \ \! 1 \! - \! m_1 \! \right\rbrace \! . $
\end{remark}
\vspace{-0.2in}

\begin{remark} From the proposed schemes, we observe that the allocation variables satisfy 
\begin{align}
\sum_{\mc S \subset [K] : \ |\mc S|=t} a_{\mc S} = t+1-\sum_{i=1}^{K}m_i, \
\sum_{\mc S \subset [K] : \ |\mc S|=t+1} a_{\mc S} =\sum_{i=1}^{K}m_i-t,
\end{align}
for $ t<\sum_{i=1}^{K}m_i \leq t+1$ and $a_{\mc S}=0$ for $|\mc S| \not\in \{t,t+1\}$. That is, the proposed placement scheme generalizes the memory sharing scheme in \cite{maddah2014fundamental}, where $a_{\mc S}=a_{\mc S^\prime} $ for $|\mc S|=|\mc S^\prime|$.
\end{remark}

%	
%--------------------------------------------------------------------------------------
\vspace{-0.2in}
\subsection{Comparison with Other Schemes with Heterogeneous Cache Sizes}
Previous work includes the layered heterogeneous caching (LHC) scheme \cite{yang2018coded,sengupta2016layered}, where the users are divided into groups and within each group the users' cache memories are divided into layers such that the users in each layer have equal cache sizes. The Maddah-Ali--Niesen caching scheme \cite{maddah2014fundamental} is applied to the fraction of the file assigned to each layer. Let $R_{\text{LHC}}(\bm m) $ denote delivery load of this scheme. We have
% In the next proposition, we show that the LHC scheme with uncoded placement is a feasible solution to optimization problem (\ref{eqn_opt1}), in turn our caching scheme always achieves a lower delivery load.
%
\begin{proposition} 
Given $K, N \geq K$ and $\bm m$, we have $ R^*_{\mathfrak{A},\mathfrak{D}}(\bm m) \leq R_{\text{LHC}}(\bm m) $.
\end{proposition}
\begin{proof} LHC scheme is a feasible (but not necessarily optimal) solution to (\ref{eqn_opt1}) shown as follows.
\underline{Grouping:} Dividing the users into disjoint groups can be represented in the placement phase by setting $a_{\mc S}=0$ for any set $\mc S$ containing two or more users from different groups. Similarly, in the delivery phase $v_{\mc T}=0$ if $\{q_1,q_2\} \subset \mc T$, $q_1$ and $q_2$ belong to distinct groups.
\underline{Layering:} Without loss of generality, assume there is one group, i.e. there are $K$ layers. Let $\alpha_l$ be the fraction of the file considered in layer $l$, and assign $a_{\mc S, l}$, $v_{\mc T, l}$, and $u^{\mc T}_{\mc S, l}$ to layer $l$, such that $a_{\mc S} =\sum\limits_{l=1}^{K} a_{\mc S, l} $, $v_{\mc T} =\sum\limits_{l=1}^{K} v_{\mc T, l} $, and $u^{\mc T}_{\mc S} =\sum\limits_{l=1}^{K} u^{\mc T}_{\mc S, l}$. Additionally, we have $\sum\limits_{\mc S \subset \{l,\dots,K\}}  \! \! \! \!   a_{\mc S,l } = \alpha_l  $, and $ \sum\limits_{\mc S \subset \{l,\dots,K\}  : \ \! \{k\} \in \mc S}    \! \! \! \! \! \! \! \! a_{\mc S,l } + \! \! \! \sum\limits_{\mc T \subset \{l,\dots,K\} : \ \! \{k\} \in \mc T} \! \! \! \! \! \!  \! \! v_{\mc T,l } \geq \alpha_l$ for $k \in \{l,\dots,K\}$. Thus, the LHC scheme is a feasible solution to (\ref{eqn_opt1}).
%$K-l+1$ is the number of users in layer $l$, $M_l-M_{l-1}$ is the memory size in layer $l$,
%Additionally, for each layer $l \in [L]$, we have  $\sum\limits_{\mc S \in 2^{[K]} \ \! : \ \! k \in \mc S } \! \!   \! a_{\mc S,l } \leq  \alpha_l m_k $, and $\sum\limits_{\mc T \in 2^{[K]}-\phi \ \! : \ \! k \in \mc T } v_{\mc T,l }  \geq \alpha_l (1-m_k)$. 
\end{proof}
The recent reference \cite{daniel2017optimization} has proposed optimizing over uncoded placement schemes $\mathfrak{A}$ with the decentralized delivery scheme in \cite{maddah2015decentralized}, i.e., the multicast signals are defined as $ X_{\mc T,\bm d}= \bar \oplus_{k \in \mc T} \tilde W_{d_k,\mc T \setminus \{k\}}$ where $v_{\mc T}= \max\limits_{k \in \mc T} a_{\mc T \setminus \{k\}}$, which limits the multicast opportunities \cite{ramakrishnan2015efficient} as illustrated in Section \ref{sec_motv}.

\begin{remark} For fixed cache contents, reference \cite{ramakrishnan2015efficient} proposed a procedure for redesigning the multicast signals, formed by XORing pieces of unequal size, in order to better utilize the side-information stored at the users. In contrast, our scheme is a centralized scheme, where we jointly optimize the cache contents and the delivery procedure which allows flexible utilization of the side-information at the users. 
%---
%The delivery scheme proposed in \cite{ramakrishnan2015efficient} can reduce the delivery load achieved by the scheme in \cite{daniel2017optimization}. However, in general, our scheme achieves a lower delivery load, because we jointly optimize over all feasible uncoded placement and linear delivery schemes. 
% For instance, applying the delivery scheme in \cite{ramakrishnan2015efficient} on the multicast signal constructed according to \cite{daniel2017optimization} in the example in Section \ref{sec_motv}, does not reduce the delivery load.
%Reference [33] attempts to improve on the decentralized delivery procedure in [44], however the gain we get is limited by the fixed cache contents. 
%The example discussed in Section III illustrates this limitation, as the procedure in [33] does not improve the performance when applied on the multicast signal constructed according to [44]. 
\end{remark}

%In particular, the joint optimization of uncoded placement and the proposed delivery scheme in order to achieve the delivery load memory trade-off under uncoded placement.
%Given a placement scheme, reference [33] proposed a procedure for improving the efficiency of the delivery phase when the multicast signals are formed by XORing pieces of unequal size [44]. Reference [33] proposed redesigning the multicast signals to better utilize the given side-information at the users, by introducing some flexibility in utilizing the side-information. Reference [33] attempts to improve on the decentralized delivery procedure in [44], however the gain we get is limited by the fixed cache contents. The example discussed in Section III illustrates this limitation, as the procedure in [33] does not improve the performance when applied on the multicast signal constructed according to [44]. 
%In contrast, our scheme is a centralized scheme, where we jointly optimize the cache contents and the delivery procedure which allows flexible utilization of the side-information at the users. In summary, the joint optimization of the two phases and the flexible utilization of side-information are both needed in order to exploit all multicast opportunities in systems with unequal caches.

Different from \cite{yang2018coded,sengupta2016layered,daniel2017optimization}, we propose a more general delivery scheme that allows flexible utilization of the side-information. Both our solution and that of \cite{daniel2017optimization} is exponential in the number of users. Notably, for systems with only two distinct cache sizes over all users, reference \cite{daniel2017optimization} has a caching scheme which is obtained by solving an optimization problem with polynomial complexity.%, by considering a subset of the uncoded placement schemes. 
% i.e., $m_k \in \{b_1,b_2\}, \forall k \in [K]$, 
%--------------------------------------------------------------------------------------
%--------------------------------------------------------------------------------------
\vspace{-0.1in}
\section{Optimizing Cache Sizes with Total Memory Budget}\label{sec_mem_opt}
In this section, we consider a centralized system where the server is connected to the $K$ users via rate limited download links of distinct capacities, as described by Fig.$\!  $ \subref*{fig_sys_model_gen}. 
\vspace{-0.2in}
\subsection{Problem Formulation}

Next, we consider the joint optimization of both the caching scheme and the users' cache sizes for given download rates $\bm C$ (see Fig.$\! $ \subref*{fig_sys_model_gen}) and normalized cache budget $m_{\text{tot}}$. More specifically, the minimum worst-case DCT, $\Theta^{*}_{\mathfrak A, \mathfrak D}(m_{\text{tot}},\bm C)$, is characterized by %For given download rates $\bm C$ (see Fig.$\! $ \subref*{fig_sys_model_gen}) and normalized cache budget $m_{\text{tot}}$, the following optimization problem identifies the minimum worst-case DCT, $\Theta^{*}_{\mathfrak A, \mathfrak D}(m_{\text{tot}},\bm C)$, the optimal memory allocation, uncoded placement, and linear delivery schemes.
\begin{subequations} \label{eqn_opt2}
\begin{align}
\textit{\textbf{O\arabic{opt_ct}}:}  \qquad  & \min_{\bm a,\bm u ,\bm v, \bm m}  
& &  \sum_{\mc T \subsetneq_{\phi} [K]} \frac{v_{\mc T}}{\min\limits_{j \in \mc T} C_j}   \\
& \ \ \text{s.t.}
& & \bm a \in \mathfrak{A}(\bm m), \ \! (\bm v, \bm u) \in \mathfrak D(\bm a), \label{eqn_O2_cnst2}\\
& & &  \sum_{k=1}^{K} m_k \leq m_{\text{tot}}, \ \! 0 \leq m_k  \leq 1, \ \forall \ \! k \in [K],
\end{align}
\end{subequations}
\stepcounter{opt_ct}
where $\mathfrak{A}(\bm m) $ is defined in (\ref{eqn_feas_alloc}) and $\mathfrak D(\bm a) $ is defined in (\ref{eqn_feas_delv}).

\vspace{-0.2in}
\subsection{Optimal Cache Sizes}

The linear program in (\ref{eqn_opt2}) characterizes the optimal memory allocation assuming uncoded placement and linear delivery schemes. For the case where $m_{\text{tot}} \leq 1$, we are able to derive a closed form expression for the optimal memory allocation, and show that the optimal solution balances between allocating larger cache memories to users with low decoding rates and uniform memory allocation. In particular, the cache memory budget $m_{\text{tot}}$ is allocated uniformly between users $\{1,\dots, q\}$, where $q$ is determined by $\bm C$ as illustrated in the following theorem.

\begin{theorem}\label{thm_optmem}
For $C_1 \leq \dots \leq C_K $ and $m_{\text{tot}} \leq 1$, the minimum worst-case delivery completion time (DCT) is given by
\vspace{-0.2in}
\begin{align}
\Theta^{*}_{\mathfrak A, \mathfrak D}(m_{\text{tot}}, \bm C)= \sum_{j=1}^{K} \dfrac{1}{C_j} - \max_{i \in [K]} \left\lbrace  \sum_{j=1}^{i} \dfrac{j \ \! m_{\text{tot}}}{i \ \! C_j} \right\rbrace,
\end{align}
and the optimal memory allocation is given by $m_1^* = \dots = m_q^* = \frac{m_{\text{tot}}}{q}$, where the user index $q = \operatorname{arg\, \! \! \max}_{i \in [K]} \left\lbrace  \sum_{j=1}^{i} \dfrac{j}{i \ \! C_j} \right\rbrace$.
\end{theorem}
Proof of Theorem \ref{thm_optmem} is provided in Appendix \ref{app_thm_optmem}. Note that if the optimal solution is not unique, i.e., $q \in \{q_1,\dots,q_L \},$ for some $ L \leq K$, then $\bm m^* = \sum\limits_{i=1}^{L} \alpha_i \big[\frac{m_{\text{tot}}}{q_i}, \dots, \frac{m_{\text{tot}}}{q_i}, 0, \dots, 0 \big],$ where $\sum\limits_{i=1}^{L} \alpha_i=1$ and $\alpha_i \geq 0$. The next proposition shows that uniform memory allocation combined with the  Maddah-Ali--Niesen caching scheme yields an upper bound on $\Theta^{*}_{\mathfrak A, \mathfrak D}(m_{\text{tot}}, \bm C)$. 
\vspace{-0.15in}
\begin{proposition}\label{prop_unf_aloc}
For $m_{\text{tot}} \in [K]$ and $C_1 \leq C_2 \leq \dots \leq C_K$, we have
\begin{align}\label{eqn_unf_dct}
\Theta^{*}_{\mathfrak A, \mathfrak D}(m_{\text{tot}}, \bm C) \leq \Theta_{\text{unif}}(m_{\text{tot}}, \bm C) = \frac{1}{\binom{K}{m_{\text{tot}}}} \sum_{j=1}^{K-m_{\text{tot}}} \frac{\binom{K-j}{m_{\text{tot}}}}{C_j}.
\end{align}
\vspace{-0.15in}
\end{proposition}
\begin{proof} Assuming $m_j=m_{\text{tot}}/K, \forall j \in [K]$, the placement phase is described by $a_{\mc S}=1/\binom{K}{m_{\text{tot}}}$ for $  |\mc S| =m_{\text{tot}}$ and zero otherwise. While, the delivery phase is defined by $v_{\mc T}=1/\binom{K}{m_{\text{tot}}} $ for $  |\mc T| =m_{\text{tot}}+1$ and $u^{\mc T}_{\mc S}=v_{\mc T}$ for $\mc S \in \{ \mc T \setminus\{j\}: j \in \mc T\}$. In turn, we have
\begin{align}
\Theta_{\text{unif}}(m_{\text{tot}}, \bm C) =\sum_{\mc T \subsetneq_{\phi} [K]} \frac{v_{\mc T}}{\min\limits_{j \in \mc T} C_j}=\!\! \sum_{\mc T \subsetneq_{\phi} [K] : \ |\mc T| =m_{\text{tot}}+1} \frac{1/\binom{K}{m_{\text{tot}}}}{\min\limits_{j \in \mc T} C_j}=\frac{1}{\binom{K}{m_{\text{tot}}}} \sum_{j=1}^{K-m_{\text{tot}}} \frac{\binom{K-j}{m_{\text{tot}}}}{C_j},
\end{align}
since $C_1 \leq C_2 \leq \dots \leq C_K$ and there are $\binom{K-j}{m_{\text{tot}}}$ sets of size $m_{\text{tot}}+1$ that include user $j$ and do not include users $\{1,2, \dots, j-1\}$. Finally, $ \Theta^{*}_{\mathfrak A, \mathfrak D}(m_{\text{tot}}, \bm C) \leq \Theta_{\text{unif}}(m_{\text{tot}}, \bm C)$, since uniform memory allocation is a feasible solution (but not necessarily optimal) to (\ref{eqn_opt2}). 
\end{proof}
%--------------------------------------------------------------------------------------
%--------------------------------------------------------------------------------------
% ######################################################################################
% ######################################################################################
\vspace{-0.1in}
\section{Numerical Results}\label{sec_numerical}
First, we provide a numerical example for the optimal caching scheme obtained from (\ref{eqn_opt1}).
\begin{example}
Consider a caching system with $K=3 $, and $\bm m=[0.4, \ 0.5, \ 0.6]$. The caching scheme obtained from (\ref{eqn_opt1}), is described as follows.\newline
\underline{Placement phase}: Every file $W^{(l)}$ is divided into six subfiles, such that $a_{\{1\}} \! = \! 7/30$, $a_{\{2\}} \! = \! 4/30$, $a_{\{3\}} \! = \! 4/30$, $a_{\{1,2\}} \! = \! 1/30$, $a_{\{1,3\}} \! = \! 4/30$, and $a_{\{2,3\}} \! = \! 10/30$. %, as shown in Fig. .
\newline 
\underline{Delivery phase}: We have the following transmissions
\begin{align*}
X_{\{1,2\}, \bm d}&= \left( W_{d_1,\{2\}}^{\{1,2\}} \bigcup W_{d_1,\{2,3\}}^{\{1,2\}} \right) \oplus \ \left( W_{d_2,\{1\}}^{\{1,2\}} \bigcup W_{d_2,\{1,3\}}^{\{1,2\}} \right), X_{\{2,3\}, \bm d}= W_{d_2,\{3\}}^{\{2,3\}} \oplus \ W_{d_3,\{2\}}^{\{2,3\}},\\
X_{\{1,3\}, \bm d}&= \left(  W_{d_1,\{3\}}^{\{1,3\}} \bigcup W_{d_1,\{2,3\}}^{\{1,3\}} \right)  \oplus \ W_{d_3,\{1\}}^{\{1,3\}}, \ 
 X_{\{1,2,3\}, \bm d}= W_{d_1,\{2,3\}}^{\{1,2,3\}} \oplus \ W_{d_2,\{1,3\}}^{\{1,2,3\}} \oplus \ W_{d_3,\{1,2\}}^{\{1,2,3\}},
\end{align*}
and the subfile sizes are as follows
\begin{enumerate}
\item $v_{\{1,2\}}=10/30 $, where $u^{\{1,2\}}_{\{2\}}=a_{\{2\}}$, $u^{\{1,2\}}_{\{2,3\}}=0.2$, $u^{\{1,2\}}_{\{1\}}=a_{\{1\}}$, and $u^{\{1,2\}}_{\{1,3\}}=0.1$.
\item $v_{\{1,3\}}=7/30 $, where $u^{\{1,3\}}_{\{3\}}=a_{\{3\}}$, $u^{\{1,3\}}_{\{2,3\}}=0.1$, and $u^{\{1,3\}}_{\{1\}}=a_{\{1\}}$.
\item $v_{\{2,3\}}=4/30 $, where $u^{\{2,3\}}_{\{3\}}=a_{\{3\}}$ and $u^{\{2,3\}}_{\{2\}}=a_{\{2\}}$.
\item $v_{\{1,2,3\}}=1/30 $, where $u^{\{1,2,3\}}_{\{2,3\}}=u^{\{1,2,3\}}_{\{1,3\}}=u^{\{1,2,3\}}_{\{1,2\}}=a_{\{1,2\}}$.
\end{enumerate}
The minimum worst-case delivery load $R^*_{\mathfrak{A}}(\bm m)=R^*_{\mathfrak{A},\mathfrak{D}}(\bm m)=22/30$.  Note that, per Remark~\ref{remark_subpack}, the required subpacketization level for the proposed scheme is $30$, i.e., the minimum packet size is given by 
\begin{equation*}
\gcd(\bm u)=\gcd\left(\dfrac{7F}{30}, \ \! \dfrac{6F}{30}, \ \! \dfrac{4F}{30}, \ \! \dfrac{3F}{30}, \ \! \dfrac{F}{30}\right)=\dfrac{F}{30},  
\end{equation*}
for $F=30n$, $n=1,2,\dots$.
\end{example}
\begin{figure}[t]
	\includegraphics[scale=0.6]{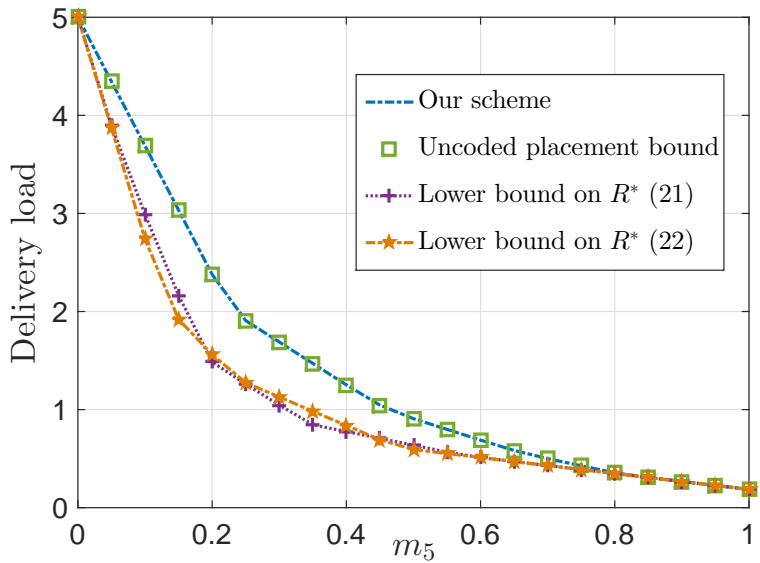}
	\centering
	\caption{Comparing $R^*_{\mathfrak{A},\mathfrak{D}}(\bm m)$ with the lower bounds in (\ref{eqn_bound_genie})-(\ref{eqn_bound_wang}), for $K=5$, and $ m_k=0.95 \ \! m_{k+1}$.}\label{fig:comp_dlv}
	\vspace{-0.3in}
\end{figure}
\begin{figure}[t]
	\includegraphics[scale=0.6]{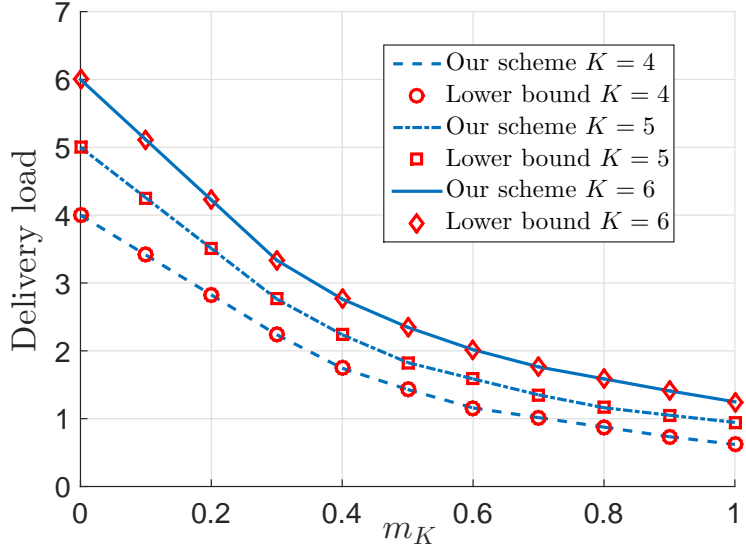}
	\centering
	\caption{Comparing $R^*_{\mathfrak{A},\mathfrak{D}}(\bm m)$ with the lower bounds on $R^*_{\mathfrak{A}}(\bm m)$, for $ m_k=0.75 \ \! m_{k+1}$.}\label{fig:R_achv_vs_lb}
	%\vspace{-0.3in}
\end{figure}
In Fig. \ref{fig:comp_dlv}, we compare the delivery load $R^*_{\mathfrak{A},\mathfrak{D}}(\bm m)$ obtained from optimization problem (\ref{eqn_opt1}), with the lower bounds on $ R^*(\bm m)$ in (\ref{eqn_bound_amiri}), (\ref{eqn_bound_wang}), and the lower bound with uncoded placement in (\ref{eqn_bound_genie}), for $N=K=5$ and $m_k=0.95 \ \! m_{k+1}$. 

From Fig. \ref{fig:comp_dlv}, we observe that $R^*_{\mathfrak{A},\mathfrak{D}}(\bm m)= R^*_{\mathfrak{A}}(\bm m)$, which is also demonstrated in Fig. \ref{fig:R_achv_vs_lb}, for $K=4,5,6$, and $m_k=0.75 \ \! m_{k+1}$. In Fig. \ref{fig:comp_dlv_others}, we compare $R^*_{\mathfrak{A},\mathfrak{D}}(\bm m)$ with the achievable delivery loads in \cite{yang2018coded,sengupta2016layered,daniel2017optimization}, for $K=5$ and $ m_k=0.75 \ \! m_{k+1}$.
\begin{figure}[t]
	\includegraphics[scale=0.6]{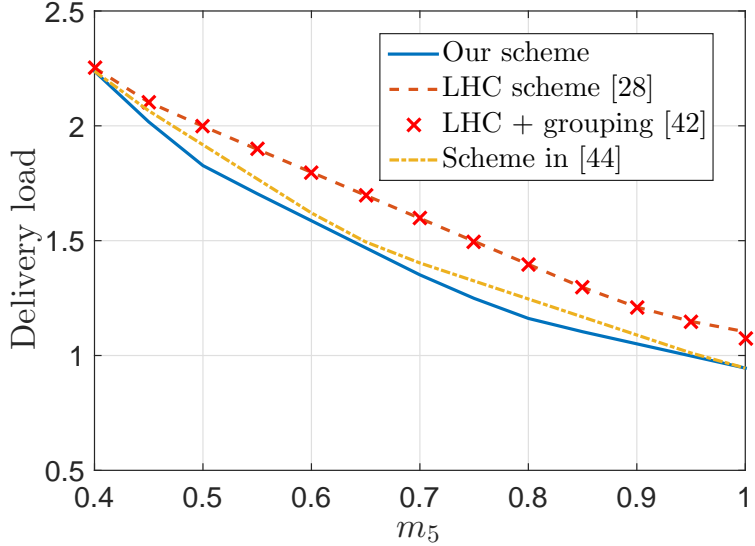}
	\centering
	\caption{Comparing $R^*_{\mathfrak{A},\mathfrak{D}}(\bm m)$ with the delivery loads achieved by the schemes in \cite{yang2018coded,sengupta2016layered,daniel2017optimization}, for $K=5$ and $ m_k=0.75 \ \! m_{k+1}$.}\label{fig:comp_dlv_others}
	\vspace{-0.2in}
\end{figure}
%
%%%%%%%%%%%%%%%%%%%%%%%%%%%%%%%%%%%%%%%%%%%%%%%%%%%%%%%%%%%%%%%%%%%%%%%%%%%%%%%%%%%%%%
%

The next example concerns with solving (\ref{eqn_opt2}) for a system with unequal download rates, and comparing the optimal memory allocation and caching scheme with the Maddah-Ali--Niesen caching scheme with uniform memory allocation. 
\begin{example}
Consider a caching system with $K=3 $, memory budget $m_{\text{tot}}=1$, and $C_1 \leq C_2 \leq C_3$, which implies $\Theta_{\text{unif}}(1, \bm C )=\frac{1}{3} \left( \frac{2}{C_1} + \frac{1}{C_2} \right),$ and $q = \operatorname{arg\, \! \! \max}_{i \in [3]} \Big\lbrace  \sum\limits_{j=1}^{i} \dfrac{j }{i \ \! C_j} \Big\rbrace$. In particular, we consider the following cases for the download rates:
\begin{enumerate}
\item For $\bm C=[0.2, \ 0.4, \ 0.5]$, we get $q=3$, hence the optimal solution is the  Maddah-Ali--Niesen caching scheme with $\bm m^* =[1/3, \ 1/3, \ 1/3]$, and we have $\Theta^{*}_{\mathfrak A, \mathfrak D}=4.1667$. %,  and the optimal caching scheme is given by 
%\begin{align*}
%&a^*_{\{1\}}=a^*_{\{2\}}=a^*_{\{3\}}=1/3, \ v^*_{\{1,2\}}=u^{ \! \!  * \{1,2\}}_{\{1\}}=u^{ \! \!  * \{1,2\}}_{\{2\}}=1/3, \\
%&v^*_{\{1,3\}}=u^{ \! \!  * \{1,3\}}_{\{1\}}=u^{ \! \!  * \{1,3\}}_{\{3\}}=1/3, \  v^*_{\{2,3\}}=u^{ \! \!  * \{2,3\}}_{\{2\}}=u^{ \! \!  * \{2,3\}}_{\{3\}}=1/3. 
%\end{align*}
%
\item For $\bm C=[0.3, \ 0.3, \ 0.6]$, we get $ q \in \{2,3\}$, i.e., the optimal solution is not unique. $\bm m^*=[\frac{\alpha}{2} \! + \! \frac{1-\alpha}{3}, \  \frac{\alpha}{2} \! + \! \frac{1-\alpha}{3}, \ \frac{1-\alpha}{3}]$, where $\alpha \in [0,1]$, and $\Theta^{*}_{\mathfrak A, \mathfrak D}=\Theta_{\text{unif}}=3.3333$.
\item For $\bm C=[0.2, \ 0.3, \ 0.6]$, we get $ q=2$. $\bm m^* =[0.5, \ 0.5, \ 0]$ and the optimal caching scheme is $a^*_{\{1\}}=a^*_{\{2\}}=0.5, $ $ v^*_{\{1,2\}}=u^{  \! \!  *  \{1,2\}}_{\{1\}}=u^{  \! \!  *  \{1,2\}}_{\{2\}}=0.5,$ $v^*_{\{3\}}=1, $
which results in $\Theta^{*}_{\mathfrak A, \mathfrak D}=4.1667 < \Theta_{\text{unif}}=4.4444.$ 
\end{enumerate}
\vspace{-0.34in}
\end{example}
\begin{figure}[t]
\includegraphics[scale=0.6]{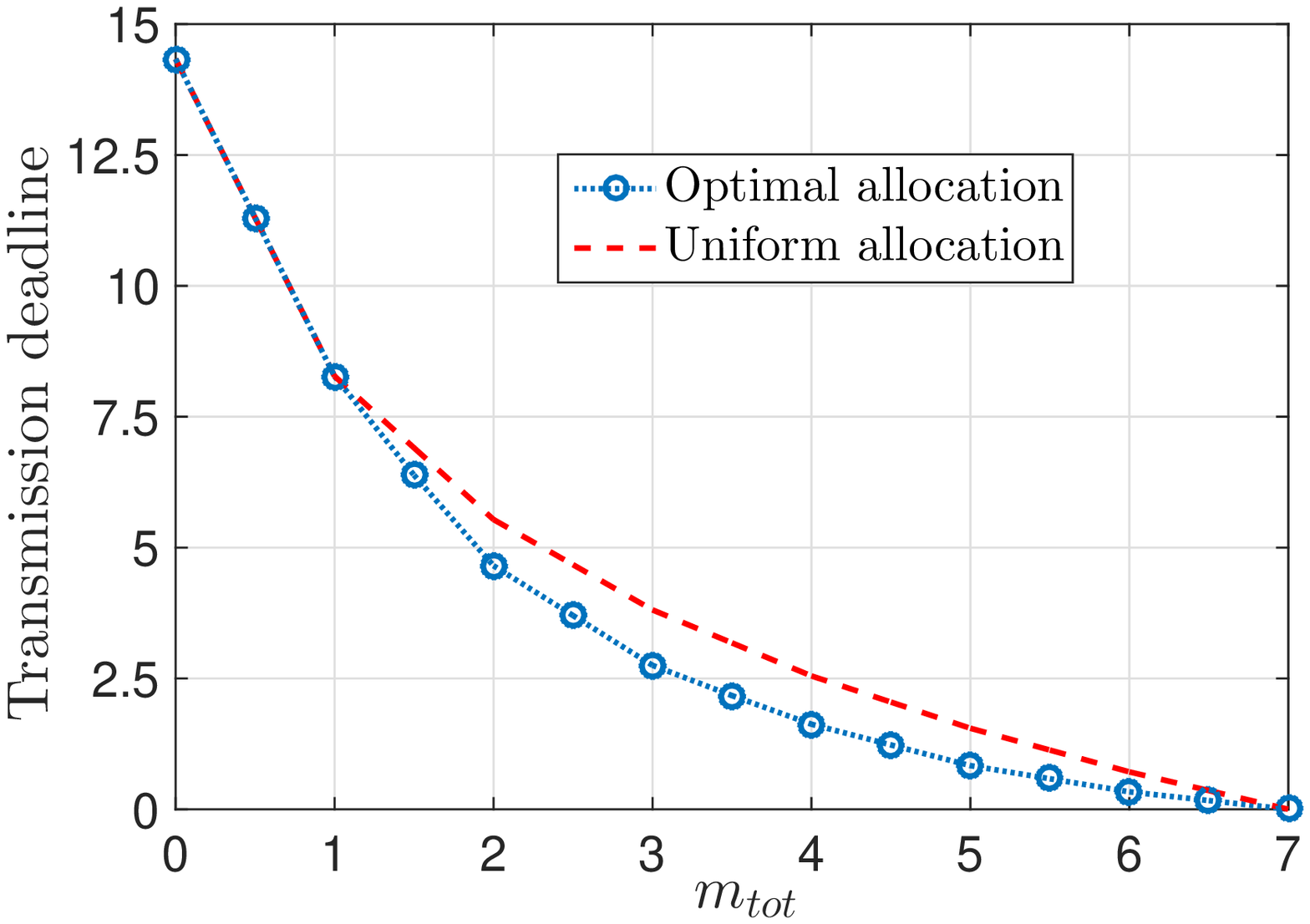}
\centering
\caption{Comparing $\Theta^{*}_{\mathfrak A, \mathfrak D}(m_{\text{tot}}, \bm C )$ and $ \Theta_{\text{unif}}(m_{\text{tot}}, \bm C )$ for $ \bm C =[0.2, \ 0.4, \ 0.6, \ 0.6, \ 0.8, \ 0.8, \ 1] $.}\label{fig_deadline_vs_mtot}
\vspace{-0.3in}
\end{figure}
In Fig. \ref{fig_deadline_vs_mtot}, we compare $ \Theta^{*}_{\mathfrak A, \mathfrak D}(m_{\text{tot}}, \bm C )$ with $ \Theta_{\text{unif}}(m_{\text{tot}}, \bm C ) $ for $K=7$, and $ \bm C =[0.2, \ 0.4, \ 0.6, $ $ 0.6, \ 0.8, \ 0.8, \ 1] $. We observe that $ \Theta^{*}_{\mathfrak A, \mathfrak D}(m_{\text{tot}}, \bm C ) \leq $ $\Theta_{\text{unif}}(m_{\text{tot}}, \bm C )$. For $m_{\text{tot}} \leq 1$, we have $ \operatorname{arg\, \! \! \max}_{i \in [K]}  \sum_{j=1}^{i} (j \ \! m_{\text{tot}})/(i \ \! C_j) = K$, which implies $\Theta^{*}_{\mathfrak A, \mathfrak D}(m_{\text{tot}}, \bm C ) = \Theta_{\text{unif}}(m_{\text{tot}}, \bm C )$.

In Fig. \subref*{fig_mem_ex1} and \subref*{fig_mem_ex2}, we show the optimal memory allocation for   
$\bm C=[0.2, \ 0.2, \ 0.2, $ $\ 0.5, \ 0.6, \ 0.7, \ 0.7]$ and $\bm C=[0.2, \ 0.2, \ 0.4, \ 0.4, \ 0.6, \ 0.7, \ 0.8]$, respectively. A general observation is that the optimal memory allocation balances the gain attained from assigning larger memories for users with weak links and the multicast gain achieved by equating the cache sizes. Consequently, in the optimal memory allocation, the users are divided into groups according to their rates, the groups that include users with low rates are assigned larger fractions of the cache memory budget, and users within each group are given equal cache sizes. These characteristics are illustrated in Fig. \subref*{fig_mem_ex1}, which shows that the users are grouped into $\mc G_1=\{1,2,3\}$ and $\mc G_2=\{4,5,6,7\}$ for all $m_{\text{tot}} \in [0,7)$. Fig. \subref*{fig_mem_ex2} shows that the grouping not only depends on the rates $\bm C$, but also on the cache memory budget $m_{\text{tot}} $. For instance, for $m_{\text{tot}}=2 $, we have two groups $\mc G_1=\{1,2\}$ and $\mc G_2=\{3,4,5,6,7\}$, however, for $m_{\text{tot}}=3 $, we have $\mc G_1=\{1,2\}$, $\mc G_2=\{3,4\}$, and $\mc G_2=\{5,6,7\}$.
\begin{figure*}[t]
	\centering
	\begin{tabular}{cc}
	\hspace{-0.2in} \subfloat[ $ C_1 =  C_2= C_3=0.2, C_4=0.5,C_5=0.6,$ and $ \qquad \qquad C_6 =C_7=0.7$.]{
				\label{fig_mem_ex1}
				\includegraphics[scale=0.46]{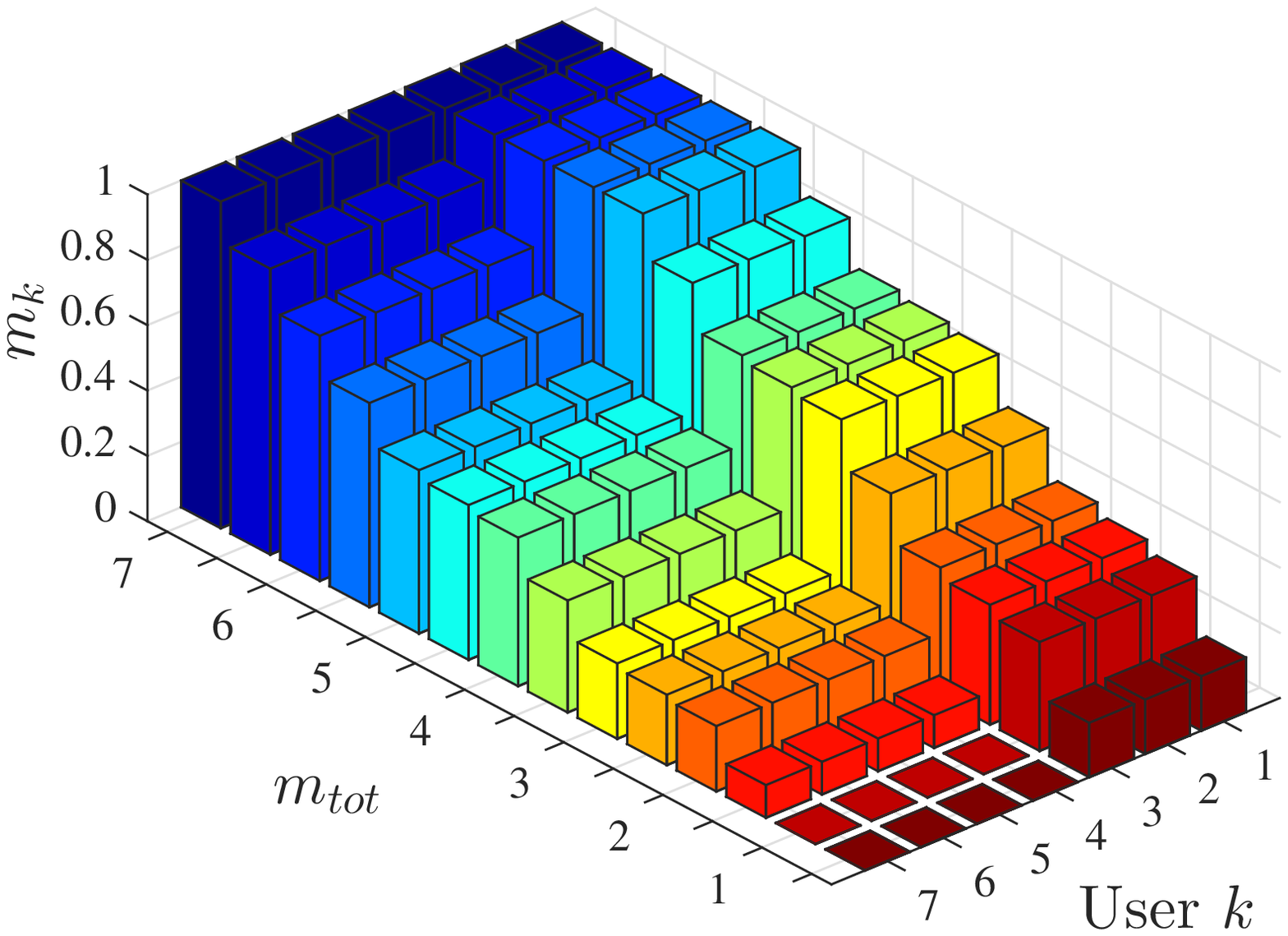} }
		
		& \hspace{-0.1in}
		\subfloat[    $  C_1 = C_2=0.2, C_3 = C_4=0.4,C_5=0.6, C_6 =0.7,$ and $C_7=0.8$.]{ \label{fig_mem_ex2}
			\includegraphics[scale=0.46]{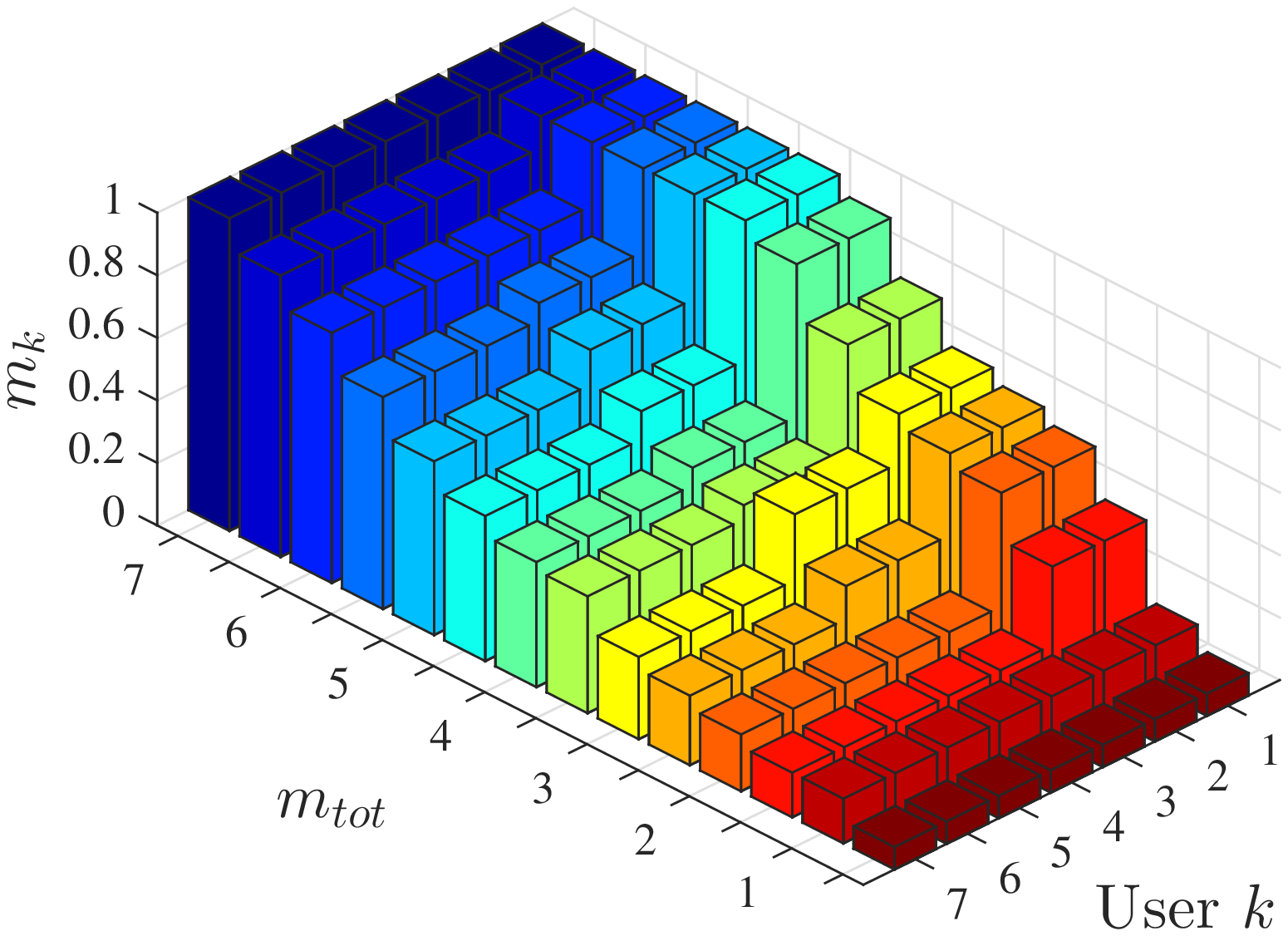}}
		\\
	\end{tabular}       
	\caption{The optimal memory allocations with different link capacities.}
	\vspace{-0.3in}
\end{figure*} 
%Finally, Fig. \ref{fig_mem_vs_c} shows the optimal memory allocation versus $C_1$ for $m_{\text{tot}}\!=\!1, C_2\!=\!0.35$ and $C_3=0.6$, we observe that the memory allocated to user $1$ is non-increasing in its link capacity.
%
%
%
%\begin{figure}[t]
%\includegraphics[scale=0.65]{figures/mem_vs_c.eps}
%\centering
%\caption{The optimal memory allocations for $m_{\text{tot}}\!=\!1, C_2\!=\!0.35$ and $C_3=0.6$.}\label{fig_mem_vs_c}
%%\vspace{-.2 in}
%\end{figure}
%
% ######################################################################################
% ######################################################################################
%\vspace{-0.2in}
\section{Conclusion}\label{sec_con}
In this paper, we have considered a downlink where the end-users are equipped with cache memories of different sizes. We have shown that the problem of minimizing the worst-case delivery load with uncoded placement and linear delivery can be modeled as a linear program. We have derived a lower bound on the worst-case delivery load with uncoded placement. We have characterized the exact delivery load memory trade-off with uncoded placement for the case where the aggregate cache size is less than or equal to the library size (small memory), the case where the aggregate cache size is greater than or equal to $K-1$ times the library size (large memory), and the three-user case for arbitrary memory size. The proposed scheme outperforms other works in the same setup \cite{yang2018coded,sengupta2016layered,daniel2017optimization}, and is numerically observed to verify the excellent performance of uncoded placement for parameters of interest.

We have also considered a system where the links between the server and the users have unequal capacities. In this scenario, the server suggests the memory sizes for cached contents along with contents to the users subject to a total memory budget, in order to minimize the delivery completion time. We have observed that the optimal solution balances between allocating larger cache sizes to the users with low link capacities and uniform memory allocation which maximizes the multicast gain. For when the total cache budget is less than the library size, we have shown that the optimal memory allocation distributes the cache memory budget uniformly over some number of users with the lowest link capacities. This number is a function of the users' link capacities.

The optimization perspective in this work provides a principled analysis of optimal caching and delivery schemes for cache-enabled networks, by translating the design elements of cache placement and delivery into structural optimization constraints. Future directions include different network topologies and systems with multiple servers and multiple libraries.

\vspace{-0.1in}
% #################################################################################
\appendices
% #################################################################################
\section{Proof of Proposition \ref{prop_side_info}}\label{app_side_info}
\vspace{-0.5in}
\begin{align}
\sum_{\mc S \subset [K] : \ \mc S^{\prime}  \subset \mc S, j \not\in \mc S} \! \! \! \! \! \! \!  a_{\mathcal S} &= a_{\mc S^{\prime}} \mathbbm{1}_{(|\mc S^{\prime}|=1)} + \sum_{\mc S \subset [K] : \ \mc S^{\prime}  \subset \mc S, j \not\in \mc S,|\mc S| \geq 2 } \! \! \! \! \! \! \!  a_{\mathcal S} \\
& \geq u^{\{j\} \cup \mc S^{\prime}}_{\mc S^{\prime}} \mathbbm{1}_{(|\mc S^{\prime}|=1)} + \sum_{\mc S \subset [K] : \! \ \mc S^{\prime}  \subset \mc S, j \not\in \mc S,|\mc \mc S| \geq 2 } \ \sum_{\mc T \subsetneq_{\phi} [K] : \ \! j \in \mc T, \mc T \cap \mc S \neq \phi,  \mc T \setminus \{j\} \subset \mc S } u^{\mc T}_{\mc S}, \label{eqn_app_sideinfo_1} \\
& \geq u^{\{j\} \cup \mc S^{\prime}}_{\mc S^{\prime}} \mathbbm{1}_{(|\mc S^{\prime}|=1)} + \sum_{\mc S \subset [K] : \! \ \mc S^{\prime}  \subset \mc S, j \not\in \mc S,|\mc \mc S| \geq 2 } \ \sum_{\mc T \subsetneq_{\phi} [K] : \ \! \{j\}  \cup \mc S^\prime \subset \mc T,  \mc T \setminus \{j\} \subset \mc S} u^{\mc T}_{\mc S}, \label{eqn_app_sideinfo_2} \\
&  = u^{\{j\} \cup \mc S^{\prime}}_{\mc S^{\prime}} \mathbbm{1}_{(|\mc S^{\prime}|=1)} + \sum_{\mc T \subsetneq_{\phi} [K] : \ \!  \{j\}  \cup \mc S^\prime \subset \mc T} \ \sum_{\mc S \subset [K] : \! \ \mc S^{\prime}  \subset \mc S, \mc T \setminus \{j\} \subset \mc S, j \not\in \mc S, |\mc \mc S| \geq 2} u^{\mc T}_{\mc S},  \label{eqn_app_sideinfo_3} \\
& = \sum_{\mc T \subsetneq_{\phi} [K] : \ \! \{j\}  \cup \mc S^\prime \subset \mc T} \ \sum_{\mc S \subset [K] : \! \ \mc T \setminus \{j\} \subset \mc S, j \not\in \mc S}u^{\mc T}_{\mc S} = \sum_{\mc T \subsetneq_{\phi} [K] : \ \! \{j\}  \cup \mc S^\prime \subset \mc T} v_{\mc T}, \label{eqn_app_sideinfo_4} 
\end{align}
where the indicator function $\mathbbm{1}_{(|\mc S^{\prime}|=1)}=1, $ if $|\mc S^{\prime}|=1 $ and zero otherwise, (\ref{eqn_app_sideinfo_1}) follows from the redundancy constraints in (\ref{eqn_KUE_redund}) and $ u^{\{j\} \cup \mc S^{\prime}}_{\mc S^{\prime}} \! \leq a_{\mc S^{\prime}} $, (\ref{eqn_app_sideinfo_2}) follows from the fact that $\mc S^\prime \subset \mc S $, and $\mc S^\prime \subset \mc T $ implies $\mc T \cap \mc S \neq \phi$. 
 By interchanging the order of summations over $\mc S$ and $\mc T$ in (\ref{eqn_app_sideinfo_2}), we get (\ref{eqn_app_sideinfo_3}), since both represent the set defined by
\begin{align}
\Big\{ (\mc T, \mc S) \big| \ \! \mc S^{\prime} \subset \mc S, \ \! j\not\in \mc S, \ \!   |\mc S|\geq 2, \ \! \{j\} \cup \mc S^{\prime} \subset \mc T, \ \! \mc T \setminus \{j\} \subset \mc S  \Big\}. 
\end{align}
The first equality in (\ref{eqn_app_sideinfo_4}) follows from the fact that $\{j\} \cup \mc S^{\prime} \subset \mc T$ and $ \mc T \setminus \{j\} \subset \mc S $ implies $\mc S^\prime \subset \mc S$, which can be proved by contradiction. More specifically, if $\mc S^\prime \not\subset \mc S$, i.e., $ \exists \ \! l \in \mc S^\prime$ and $l \not\in \mc S$, then $\{j\} \cup \mc S^\prime \subset \mc T $ implies $l \in \mc T \! \setminus \! \{j\}$. This contradicts $\mc T \! \setminus \! \{j\} \subset \mc S $, since $l \in \mc T \! \setminus \! \{j\}$ and $l \not\in \mc S $. 

The last equality follows from the structural constraints in (\ref{eqn_KUE_struct}).

%Assume $\mc S^\prime \not\subset \mc S$, i.e., $ \exists \ \! l \in \mc S^\prime$ and $l \not\in \mc S$. Then, $\{j\} \cup \mc S^\prime \subset \mc T $ implies $l \in \mc T \! \setminus \! \{j\}$. This contradicts $\mc T \! \setminus \! \{j\} \subset \mc S $, since $l \in \mc T \! \setminus \! \{j\}$ and $l \not\in \mc S $.
%Assume $\mc T \! \setminus \! \{j\} \not\subset \mc S $, i.e., $ \exists \ \! l \not \in \{j\} \cup \mc S^\prime$ such that $\mc T \! \setminus \! \{j\} \! = \! \{l\} \cup \mc S^\prime$, and $l \not\in \mc S $. Then, $|\mc T| \! = |\mc T \! \setminus \! \{j\}|+ |\{j\}| = |\mc S^\prime|\!+ \! 2$, which contradicts the fact that $|\mc T| \! \leq |\mc S^\prime|\!+ \! 1$.

% #################################################################################
\section{Proof of Theorem \ref{thm_bound_genie}: Lower bound with uncoded placement}\label{app_thm_genie}

References \cite{wan2016optimality,wan2017novel} have shown that the delivery phase is equivalent to an index-coding problem and the delivery load is lower bounded by the acyclic index-coding bound \cite[Corollary 1]{arbabjolfaei2013capacity}. Reference \cite{yu2016exact} has proposed an alternative proof for the uncoded placement bound \cite{wan2016optimality,wan2017novel} using a genie-aided approach. For ease of exposition, we will follow the genie-aided approach \cite{yu2016exact}. We consider a virtual user whose cache memory is populated by a genie. For any permutation of the users $[q_1, \dots, q_K]$, the virtual user caches the file pieces stored at user $q_j$ excluding files requested by $\{ q_1,\dots,q_{j-1} \}$ for $j \in [K] $, i.e., the virtual users cache content is given by
\begin{align}
Z_{vir} = \bigcup_{j=1}^{K}  \bigcup_{ l \in [N]\setminus \{ d_{q_1},\dots,d_{q_{j-1}} \} } \ \bigcup_{\mc S \subset [K] : \ \! \{q_j\} \in \mc S, \{q_1,\dots,q_{j-1}\} \cap \mc S=\phi } \tilde W_{l,\mc S}.
\end{align}
Using the virtual user cache content and the transmitted signals, we can decode all the requested files. Additionally, for any uncoded placement $ \bm a \in \mathfrak A (\bm m)$, the worst-case delivery load $ R^*_{\mathfrak A}(\bm m, \bm a) $ satisfies  \cite{yu2016exact} $R^*_{\mathfrak A}(\bm m, \bm a) \geq \sum\limits_{j=1}^{K} \sum\limits_{\mc S \subset [K] : \{ q_1,\dots,q_{j} \} \cap \mc S=\phi } \! \! \! \! a_{\mc S} , \ \forall \bm q \in \mc P_{[K]}, $ 
%\begin{align} \label{eqn_uncod_bound_a}
%R^*_{\mathfrak A}(\bm m, \bm a) \geq \sum_{j=1}^{K} \sum_{\mc S \subset [K] : \{ q_1,\dots,q_{j} \} \cap \mc S=\phi } a_{\mc S} , \ \forall \bm q \in \mc P_{[K]}, 
%\end{align} 
where $\mc P_{[K]} $ is the set of all permutations of $[K]$. Hence, by taking the convex combination over all possible permutations of the users, we get
\vspace{-0.2in}    
\begin{align}
R^*_{\mathfrak A}(\bm m, \bm a) &\geq \sum_{\bm q \in  \mc P_{[K]} } \alpha_{\bm q} \left(\sum_{j=1}^{K} \sum_{\mc S \subset [K] : \{ q_1,\dots,q_{j} \} \cap \mc S=\phi } a_{\mc S}\right)  , \\
&= \sum_{\bm q \in  \mc P_{[K]} } \alpha_{\bm q} \left(K a_{\phi} + \sum_{j=1}^{K-1} j \sum_{\mc S \subset [K] : \{ q_1,\dots,q_{j} \} \cap \mc S=\phi, \ \! q_{j+1} \in \mc S } a_{\mc S}\right), 
\end{align}
where $\sum\limits_{\bm q \in  \mc P_{[K]} } \alpha_{\bm q} =1$, and $\alpha_{\bm q} \geq 0 , \forall \bm q \in  \mc P_{[K]} $. By rearranging the summations, we get
\vspace{-0.2in} 
\begin{align}
R^*_{\mathfrak A}(\bm m, \bm a) \geq \sum_{\mc S \subset [K]} \gamma_{\mc S} \ \! a_{\mc S},
\end{align}
where $\gamma_{\mc S}$ is given by (\ref{eqn_gamma_S}).
%\begin{align}
%\gamma_{\mc S} \triangleq \begin{cases} K, \text{ for } \mc S =\phi,\\
%0, \text{ for } \mc S =[K], \\
%\sum\limits_{j=1}^{K-|\mc S|} \sum\limits_{\substack{\bm q \in  \mc P_{[K]}   : \ \!  q_{j+1}  \in \mc S, \\ \{ q_1,\dots,q_{j} \} \cap \mc S=\phi } } j \ \! \alpha_{\bm q}, \text{ otherwise}.
%\end{cases}
%\end{align}
%
%$ 0 \leq \gamma_{\phi} \leq K$, $\gamma_{[K]}=0$, and for the remaining $\mc S \in 2^{[K]}$, we have 
%\begin{align}
%0 \leq \gamma_{\mc S} \leq \sum_{j=1}^{K-|\mc S|} \sum_{\bm q \in  \mc P_K  \ \! : \{ q_1,\dots,q_{j} \} \cap \mc S=\phi, \ \! \{q_{j+1}\} \in \mc S} j \ \! \alpha_{\bm q}.
%\end{align}
Therefore, we have %the worst-case delivery load with uncoded placement $ R^*_{\mathfrak A}(\bm m) $ satisfies 
\vspace{-0.2in} 
\begin{align}\label{eqn_app_R_pr}
R^*_{\mathfrak A}(\bm m) \geq \min_{\bm a \in \mathfrak A (\bm m)} \sum_{\mc S \subset [K]} \gamma_{\mc S} \ \! a_{\mc S}, 
\end{align}
Furthermore, the dual of the linear program in (\ref{eqn_app_R_pr}) is given by 
\begin{subequations} \label{eqn_opt6}
	\begin{align}
%	\textit{\textbf{O\arabic{opt_ct}}:}  \qquad  
	& \max_{\lambda_{k} \geq 0,\lambda_{0}}  
	& &  -\lambda_0-\sum_{k=1}^{K} m_k \ \! \lambda_k \\
	& \text{subject to}
	& & \lambda_0+ \sum_{k \in \mc S} \lambda_k + \gamma_{\mc S} \geq 0, \ \! \forall \mc S \subset [K],
	\end{align}
\end{subequations}
%\stepcounter{opt_ct}
where $\lambda_0 $ and $\lambda_k$ are the dual variables associated with $ \sum\limits_{\mc S \subset [K]} a_{\mc S}=1 $, and $\sum\limits_{\mc S \subset [K]: k \in \mc S} a_{\mc S}\leq m_k $, respectively. By taking the maximum over all convex combination, we obtain (\ref{eqn_bound_genie}).
\vspace{-0.15in}
\section{Proof of Theorem \ref{thm_3ue} : $R^*_{\mathfrak{A}}(\bm m)$ for $K=3$}\label{app_thm_3ue}

\subsection{Region \rm{I}: $\sum_{j=1}^{3} m_j \leq 1$ (Follows from Theorem \ref{thm_spchcase})}
\subsection{Region \rm{II}: $ 1 < \sum_{j=1}^{3} m_j \leq 2$, $m_3 < m_2 +3 m_1 -1$, and $m_3 < 2 (1- m_2)$ }
\textbf{Achievability:} The caching scheme defined in Table \ref{tab_reg2} achieves 
$R^*_{\mathfrak{A}}(\bm m) \! = \! \dfrac{5}{3} \! - \!  m_1 \! - \dfrac{2 \ \! m_2}{3} \! -  \! \dfrac{m_3}{3} $.

\begin{table}[h]
\centering
\caption{Optimal caching scheme for Region \rm{II}}\label{tab_reg2}
\begin{tabular}{|l|l|}
  \hline
  Placement scheme & Delivery scheme\\
  \hline	  
   $a_{\{1\}}=(2+m_2-m_3)/3-m_1$,  &  $v_{\{1,2\}} = (2+ 2 m_3-2 m_2)/3-m_1$, $u^{\{1,2\}}_{\{1,3\}} = m_3-m_2$, \\
    $ a_{\{2\}}=a_{\{3\}}=(2-2 \ \! m_2-m_3)/3$, &  $v_{\{1,3\}} = (2+ m_2-m_3)/3-m_1$, $u^{\{1,3\}}_{\{2,3\}} = m_2-m_1$,\\
   $a_{\{1,2\}}=m_1-(m_3+1-m_2)/3$, &  $v_{\{2,3\}} = u^{\{2,3\}}_{\{2\}} = u^{\{2,3\}}_{\{3\}} = (2- 2 m_2-m_3)/3 $, \\ 
	$a_{\{1,3\}}=m_1-(2 m_2+1-2 m_3)/3$, & $v_{\{1,2,3\}} = m_1 + (m_2-m_3-1)/3$, $u^{\{1,2\}}_{\{2,3\}} = m_3-m_1$, \\ 
	$a_{\{2,3\}}=(4 m_2+2 m_3 -1)/3-m_1$. & $u^{\{1,2\}}_{\{1\}} = u^{\{1,3\}}_{\{1\}} = a_{\{1\}}$, $u^{\{1,3\}}_{\{3\}} = a_{\{3\}}$, $u^{\{1,2\}}_{\{2\}} = a_{\{2\}}$. \\
\hline
\end{tabular}
\vspace{-0.2in}
\end{table} 

\textbf{Converse:} For any $\bm a \in \mathfrak{A}(\bm m)$ and a permutation $\bm q$, we have
\vspace{-0.2in} 
\begin{align}
R^*_{\mathfrak A}(\bm m, \bm a) &\geq 3 a_{\phi}+2 a_{\{3\}} + a_{\{2\}}+ a_{\{2,3\}}, \text{ for } \bm q =[1, \ 2, \ 3], \label{eqn_app_3UE_1} \\
R^*_{\mathfrak A}(\bm m, \bm a) &\geq 3 a_{\phi}+2 a_{\{1\}} + a_{\{3\}}+ a_{\{1,3\}}, \text{ for } \bm q =[2, \ 3, \ 1], \label{eqn_app_3UE_2} \\
R^*_{\mathfrak A}(\bm m, \bm a) &\geq 3 a_{\phi}+2 a_{\{2\}} + a_{\{3\}}+ a_{\{2,3\}}, \text{ for } \bm q =[1, \ 3, \ 2]. \label{eqn_app_3UE_3}
\end{align}
Hence, by taking the average of (\ref{eqn_app_3UE_1})-(\ref{eqn_app_3UE_3}), we get
\begin{align}
& \! \! \! \! R^*_{\mathfrak A}(\bm m, \bm a) \geq 3 a_{\phi}+ \dfrac{2 a_{\{1\}}}{3} + a_{\{2\}}+ \dfrac{4 a_{\{3\}}}{3} + \dfrac{2 a_{\{2,3\}}}{3}+\dfrac{a_{\{1,3\}}}{3}, \\
&\! \! \! \! R^*_{\mathfrak A}(\bm m) \! \geq \! \min_{\bm a \in \mathfrak A (\bm m)} \! \left\lbrace \! \dfrac{5 a_{\phi}}{3} + \dfrac{2 a_{\{1\}}}{3} + a_{\{2\}}+ \dfrac{4 a_{\{3\}}}{3} + \dfrac{2 a_{\{2,3\}}}{3}+\dfrac{a_{\{1,3\}}}{3} \! \right\rbrace \! = \! \dfrac{5}{3} \! - \!  m_1 \! - \dfrac{2 \ \! m_2}{3} \! -  \! \dfrac{m_3}{3}, 
\end{align}
% &\geq \dfrac{5 a_{\phi}}{3} + \dfrac{2 a_{\{1\}}}{3} + a_{\{2\}}+ \dfrac{4 a_{\{3\}}}{3} + \dfrac{2 a_{\{2,3\}}}{3}+\dfrac{a_{\{1,3\}}}{3}, 
%and consequently
%\begin{align}
%\! \! \! \! R^*_{\mathfrak A}(\bm m) \! &\geq \! \min_{\bm a \in \mathfrak A (\bm m)} \! \left\lbrace \! \dfrac{5 a_{\phi}}{3} + \dfrac{2 a_{\{1\}}}{3} + a_{\{2\}}+ \dfrac{4 a_{\{3\}}}{3} + \dfrac{2 a_{\{2,3\}}}{3}+\dfrac{a_{\{1,3\}}}{3} \! \right\rbrace \! = \! \dfrac{5}{3} \! - \!  m_1 \! - \dfrac{2 \ \! m_2}{3} \! -  \! \dfrac{m_3}{3}, 
%\end{align}
which is obtained by solving the dual linear program, as in Appendix \ref{app_thm_genie}.
%-----------------------------------------------------------------------------------
%\vspace{-0.2in}
\vspace{-0.25in}
\begin{table}[t]
\centering
\caption{Optimal caching scheme for Region \rm{III}}\label{tab_reg3}
\begin{tabular}{|l|l|l|}
  \hline
  Conditions & Placement scheme & Delivery scheme\\
  \hline	  
  $m_1 \leq 1/3$,  & $a_{\{1\}}=m_1$,  &  $v_{\{1\}} = 1-3 m_1$ , $v_{\{2\}} = m_3-m_2$,\\
   $m_1+m_3 <1 $ & $a_{\{2\}}=1-(m_1+m_3) $, &  $v_{\{1,2\}} = u^{\{1,2\}}_{\{1\}} = u^{\{1,2\}}_{\{2\}} =m_1 $, \\
   & $a_{\{3\}}=1-(m_1+m_2) $, &  $v_{\{1,3\}} = u^{\{1,3\}}_{\{1\}} = u^{\{1,3\}}_{\{3\}} =m_1$, \\ 
	& $a_{\{2,3\}}=\sum_{j=1}^{3} m_j-1 $. & $v_{\{2,3\}} = u^{\{2,3\}}_{\{2\}} = u^{\{2,3\}}_{\{3\}} = 1-(m_1+m_3) $. \\ 
\hline
  $m_1 >  1/3$,  & $a_{\{1\}}=1-2 m_1 $,  &  $v_{\{2\}} = 1 +m_3 - 3 m_1-m_2$, \\
   $m_3 < 2 m_1 $ & $a_{\{2\}}=2 m_1-m_3 $, & $v_{\{1,2\}} = u^{\{1,2\}}_{\{1\}}+u^{\{1,2\}}_{\{1,3\}} =m_1 $, $u^{\{1,2\}}_{\{2\}}=a_{\{2\}}$, \\
   & $a_{\{3\}}=1-(m_1+m_2) $, & $v_{\{1,3\}} = u^{\{1,3\}}_{\{1\}}=1-2 m_1 $ , $u^{\{1,3\}}_{\{3\}}=a_{\{3\}}$\\ 
	& $a_{\{1,3\}}=3 m_1-1$, & $v_{\{2,3\}} = u^{\{2,3\}}_{\{2\}}= u^{\{2,3\}}_{\{3\}}=2 m_1 - m_3$, \\
	& $a_{\{2,3\}}=m_2+m_3-2 m_1$. & $u^{\{1,2\}}_{\{2,3\}}=m_3-m_1$, $u^{\{1,3\}}_{\{2,3\}}=m_2-m_1$.\\ 
  \hline
  $m_1+m_3 \geq 1$,  & $a_{\{1\}}=1-m_3$ ,  &  $v_{\{1\}} = m_3-2 m_1$ , $v_{\{2\}} = 1-(m_1+m_2)$ , \\
   $m_3 \geq 2 m_1 $ & $a_{\{3\}}=1-(m_1+m_2)$, &  $v_{\{1,2\}} = u^{\{1,2\}}_{\{1\}}+u^{\{1,2\}}_{\{1,3\}} =u^{\{1,2\}}_{\{2,3\}} =m_1 $, \\
   & $a_{\{1,3\}}=m_1+m_3-1 $, &  $v_{\{1,3\}} = u^{\{1,3\}}_{\{1\}}=1-m_3 $,\\ 
	& $a_{\{2,3\}}=m_2 $. & $u^{\{1,3\}}_{\{3\}} + u^{\{1,3\}}_{\{2,3\}} =1-m_3$. \\
  \hline
\end{tabular}
\vspace{-0.2in}
\end{table} 

\subsection{Region \rm{III}: $ 1 < \sum_{j=1}^{3} m_j \leq 2$, $m_1 + m_2 < 1 $, and $m_3 \geq m_2 +3 m_1 -1$ }
\textbf{Achievability:} There are multiple caching schemes that achieve $R^*_{\mathfrak{A},\mathfrak{D}}(\bm m) \!  = \! 2 \! - \!  2 m_1 \! - m_2$. In particular, we consider caching schemes that satisfy 
\vspace{-0.15in}
\begin{align}\label{eqn_app_3ue_rIII_1}
&a_{\{1\}} \! + \! a_{\{1,3\}} \! = \! m_1, \ \! a_{\{2\}} \! + \! a_{\{2,3\}} \! = \! m_2,  
a_{\{3\}} \! = \! 1 \! - \! m_1 \! - \! m_2, \! \ a_{\{1,3\}} \! + \! a_{\{2,3\}} \! = \!m_1 \! + \! m_2 \! + \! m_3 \! - \! 1, \\
&v_{\{1,2\}} =m_1, \ v_{\{1,3\}}=a_{\{1\}}, \ v_{\{2,3\}}=a_{\{2\}}, \ v_{\{1,3\}}+v_{\{2,3\}}=1-m_3\\
&v_{\{1\}}+v_{\{1,3\}}=1-2m_1, \ 
v_{\{2\}}+v_{\{2,3\}}=1-m_1 -m_2. \label{eqn_app_3ue_rIII_2}
\end{align}
%\vspace{-0.1in}
In Table \ref{tab_reg3}, we provide one feasible solution to (\ref{eqn_app_3ue_rIII_1})-(\ref{eqn_app_3ue_rIII_2}).

\textbf{Converse:} For any $\bm a \in \mathfrak{A}(\bm m)$ and $\bm q = [1, \ 2, \ 3]$, we have
\vspace{-0.1in} 
\begin{align}
R^*_{\mathfrak A}(\bm m, \bm a) &\geq 3 a_{\phi}+2 a_{\{3\}} + a_{\{2\}}+ a_{\{2,3\}}
 \geq 2 a_{\phi}+2 a_{\{3\}} + a_{\{2\}}+ a_{\{2,3\}}, \label{eqn_app_regIII} \\
 R^*_{\mathfrak A}(\bm m) \! &\geq \! \min_{\bm a \in \mathfrak A (\bm m)} \! \left\lbrace 2 a_{\phi}+2 a_{\{3\}} + a_{\{2\}}+ a_{\{2,3\}} \right\rbrace \! = \! 2 \! - \!  2 m_1 \! - m_2.
\end{align}
%
%
%-----------------------------------------------------------------------------------
\vspace{-0.5in}
\subsection{Region \rm{IV}:  $ m_1+m_2 > 1$, and $m_3 \geq 2 (1- m_2)$ }
\begin{table}[t]
\centering
\caption{Optimal caching scheme for Region \rm{IV}}\label{tab_reg4}
\begin{tabular}{|l|l|l|}
  \hline
  Conditions & Placement scheme & Delivery scheme\\
  \hline	  
  $m_1+m_2+m_3 \geq 2$,  & $a_{\{1,2\}}=1-m_3 $,  &  $v_{\{1,2\}} = u^{\{1,2\}}_{\{1,3\}} = u^{\{1,2\}}_{\{2,3\}} =m_3-m_1 $, \\
   $1+m_1 \geq m_2 +m_3 $ & $a_{\{1,3\}}=1-m_2 $, & $v_{\{1,3\}} = u^{\{1,3\}}_{\{1,2\}} = u^{\{1,3\}}_{\{2,3\}} =m_2-m_1 $, \\
   & $a_{\{2,3\}}=1-m_1 $, & $v_{\{1,2,3\}} = 1+m_1 -(m_2+m_3) $, \\ 
	& $a_{\{1,2,3\}}=\sum_{j=1}^{3} m_j-2 $. & \\ 
\hline
  $m_1+m_2+m_3 \geq 2$,  & $a_{\{1,2\}}=1-m_3 $,  &  $v_{\{1\}} = m_2 +m_3 - (1+m_1)$, \\
   $1+m_1 < m_2 +m_3 $ & $a_{\{1,3\}}=1-m_2 $, & $v_{\{1,2\}} = u^{\{1,2\}}_{\{1,3\}} = u^{\{1,2\}}_{\{2,3\}} =1-m_2 $, \\
   & $a_{\{2,3\}}=1-m_1 $, & $v_{\{1,3\}} = u^{\{1,3\}}_{\{1,2\}} = u^{\{1,3\}}_{\{2,3\}} =1-m_3 $.\\ 
	& $a_{\{1,2,3\}}=\sum_{j=1}^{3} m_j-2 $. & \\ 
  \hline
  $m_1+m_2+m_3 < 2$,  & $a_{\{1\}}=2-\sum_{j=1}^{3} m_j$ ,  &  $v_{\{1,2\}} =  u^{\{1,2\}}_{\{2,3\}} =m_3-m_1 $, $u^{\{1,2\}}_{\{1\}} = a_{\{1\}}$, \\
   $1+m_1 \geq m_2 +m_3 $ & $a_{\{1,2\}}=m_1+m_2-1 $, &  $v_{\{1,3\}} = u^{\{1,3\}}_{\{2,3\}} =m_2-m_1 $, $u^{\{1,3\}}_{\{1\}} = a_{\{1\}}$, \\
   & $a_{\{1,3\}}=m_1+m_3-1 $, &  $v_{\{1,2,3\}} = 1+m_1  -(m_2+m_3) $, \\ 
	& $a_{\{2,3\}}=1-m_1 $. & $u^{\{1,2\}}_{\{1,3\}} \! = \! m_2 \! + \! 2 m_3 \! - \! 2$, $ u^{\{1,3\}}_{\{1,2\}} \! = \! 2 m_2 \! + \! m_3 \! - \! 2 $.\\
  \hline
  $m_1+m_2+m_3 < 2$,  & $a_{\{1\}}=2-\sum_{j=1}^{3} m_j$ ,  & $v_{\{1\}} = m_2+m_3 -(1+m_1) $,  \\
   $1+m_1 < m_2 +m_3 $ & $a_{\{1,2\}}=m_1+m_2-1 $, & $v_{\{1,2\}} =  u^{\{1,2\}}_{\{2,3\}} =1-m_2 $,  $u^{\{1,2\}}_{\{1\}} =  a_{\{1\}}$,\\
   & $a_{\{1,3\}}=m_1+m_3-1 $, &  $v_{\{1,3\}} = u^{\{1,3\}}_{\{2,3\}} =1-m_3 $, $u^{\{1,3\}}_{\{1\}} = a_{\{1\}}$, \\ 
	& $a_{\{2,3\}}=1-m_1 $. &  $u^{\{1,2\}}_{\{1,3\}} \! = \! m_1 + m_3 -1$, $ u^{\{1,3\}}_{\{1,2\}} \! = \! m_1 + m_2 -1$.\\
	\hline
\end{tabular}
\vspace{-0.3in}
\end{table} 
%
%$ m_1+m_2 > 1$, and $m_3 \geq 2 (1- m_2)$
%
\textbf{Achievability:} There are multiple caching schemes that achieve $R^*_{\mathfrak{A},\mathfrak{D}}(\bm m) \!  = \! 1 \! - \!  m_1 $. In particular, we consider caching schemes that satisfy 
\vspace{-0.2in}
\begin{align}
&a_{\{1\}}-a_{\{1,2,3\}} = 2- (m_1 +m_2 + m_3), \ a_{\{2,3\}}=1-m_1, \label{eqn_app_3ue_rIV_1} \\
&a_{\{1,2\}}+a_{\{1,2,3\}} = m_1+m_2-1, \
a_{\{1,3\}}+a_{\{1,2,3\}} = m_1+m_3-1, \\
&v_{\{1\}} \! + \! v_{\{1,2\}} \! + \! v_{\{1,3\}} \! + \! v_{\{1,2,3\}} \! = \! 1 \! - \!m_1, \ \! v_{\{1,2\}} \! + \! v_{\{1,2,3\}} \! = \! 1 \! - \! m_2, \ \! v_{\{1,3\}} \! + \! v_{\{1,2,3\}} \! = \! 1 \! - \! m_3. \label{eqn_app_3ue_rIV_2}
\end{align}
In Table \ref{tab_reg4}, we provide one feasible solution to (\ref{eqn_app_3ue_rIV_1})-(\ref{eqn_app_3ue_rIV_2}).

\textbf{Converse:} From the cut-set bound in (\ref{eqn_bound_wang}), we have $R^*(\bm m) \geq  \! 1 \! - \!  m_1$. 

%From (\ref{eqn_app_regIII}), we get  
%\begin{align}
%R^*_{\mathfrak A}(\bm m) \! \geq \! \min_{\bm a \in \mathfrak A (\bm m)} \! \left\lbrace a_{\phi}+ a_{\{3\}} + a_{\{2\}}+ a_{\{2,3\}} \right\rbrace \! = \! 1 \! - \!  m_1.
%\end{align}

% #################################################################################

%\vspace{-0.2in}
\section{Proof of Theorem \ref{thm_optmem}: Optimal Cache Sizes }\label{app_thm_optmem} 
 
First, for a given memory allocation with $\sum_{k=1}^{K} m_k \leq 1$, we have the following Lemma. %shows that the optimal caching scheme that minimizes the worst-case delivery completion time is similar to the one given in Theorem \ref{thm_spchcase}.
\begin{lemma}\label{lemma_optcach}
For $C_1 \leq \dots \leq C_K$ and memory allocation $\bm m$ satisfying $\sum_{k=1}^{K} m_k \leq 1$, the optimal caching scheme for (\ref{eqn_opt2}) is given by $a^*_{\{j\}}=m_j$, $v^*_{\{i,j\}}=u^{ \! \! * {\{i,j\}}}_{\{i\}}= u^{ \! \! * {\{i,j\}}}_{\{j\}}=\min \{ a^*_{\{i\}},a^*_{\{j\}} \}$, and $v^*_{\{j\}}=1-m_j-\sum\limits_{i=1,i \neq j}^{K} \min \{m_i, m_j\}$.
\end{lemma}
\begin{proof} By combining (\ref{eqn_KUE_delv}) with (\ref{eqn_feas_alloc}), dividing it by $C_k$, and summing over $k$, we get
\begin{align}
&\sum_{k=1}^{K} \sum_{\mc T \subsetneq_{\phi} [K] : k \in \mc T} \frac{v_{\mc T}}{C_k} \geq \sum_{k=1}^{K} \frac{1-m_k}{C_k}, \\
&\sum_{k=1}^{K} \frac{v_{\{k\}}}{C_k} \geq \sum_{k=1}^{K} \frac{1-m_k}{C_k} - \! \! \! \! \! \sum_{\mc T \subsetneq_{\phi} [K] : |\mc T| \geq 2} \ \sum_{j \in \mc T} \frac{v_{\mc T}}{C_j}.
\end{align}
\vspace{-0.15in}
Therefore, we get the lower bound
\begin{align}
\Theta_{\mathfrak A, \mathfrak D}(m_{\text{tot}}, \bm C) \geq  \sum_{k=1}^{K} \frac{1-m_k}{C_k} - \! \! \! \! \! \! \sum_{\mc T \subsetneq_{\phi} [K] : |\mc T| \geq 2} \! \! \! \! \! \! v_{\mc T} \bigg( \frac{-1}{\min_{i \in \mc T} C_i}+ \sum_{j \in \mc T} \frac{1}{C_j} \bigg).
\end{align}
Additionally, for $C_1 \leq \dots \leq C_K $, we have 
\begin{align}
\Theta_{\mathfrak A, \mathfrak D}(m_{\text{tot}}, \bm C) &\geq  \sum_{k=1}^{K} \dfrac{1-m_k}{C_k} \! - \! \sum_{i=1}^{K-1} \sum_{j=i+1}^{K} \! \dfrac{1}{C_j} \sum_{\mc T \subsetneq_{\phi} [K] : \{i,j\} \subset \mc T} \! \! \! \! \! \! \! v_{\mc T} , \\
& \geq \sum_{k=1}^{K} \dfrac{1-m_k}{C_k} - \sum_{i=1}^{K-1} \sum_{j=i+1}^{K} \dfrac{\min\{m_i,m_j\}}{C_j},
\end{align}
where the last inequality follows from the fact that the multicast transmissions that include users $\{i,j\}$ are limited by the side-information stored at each of them, which is upper bounded by the cache memory size, i.e., $\sum_{\mc T \subsetneq_{\phi} [K]: \ \! \{i,j\} \subset \mc T}  v_{\mc T} \leq \min\{m_i,m_j\} $. Finally, for $\sum_{i=1}^{K} m_i  \leq 1 $, $a^*_{\{j\}}=m_j$, $v^*_{\{j\}}=1-m_j-\sum_{i=1,i \neq j}^{K} \min \{m_i, m_j\}$, and $v^*_{\{i,j\}}=u^{ \! \! * {\{i,j\}}}_{\{i\}}= u^{\! \! * {\{i,j\}}}_{\{j\}}=\min \{ a^*_{\{i\}},a^*_{\{j\}} \}$, is a feasible solution to (\ref{eqn_opt2}) that achieves the lower bound.
\end{proof}

Now, using Lemma \ref{lemma_optcach}, we can simplify (\ref{eqn_opt2}) to
\begin{subequations} \label{eqn_opt3}
\begin{align}
& \min_{ \bm m}  
& & \sum_{k=1}^{K} \dfrac{1-m_k}{C_k} - \sum_{i=1}^{K-1} \sum_{j=i+1}^{K} \dfrac{\min\{m_i,m_j\}}{C_j}\\
& \ \text{s.t.}
& &  \sum_{k=1}^{K} m_k \leq m_{\text{tot}}, \ \! 0 \leq m_k  \leq 1, \ \forall \ \! k \in [K]. 
\end{align}
\end{subequations}
Next, we show that the optimal memory allocation from (\ref{eqn_opt3}) satisfies $m_1^* \geq m_2^* \geq \dots \geq m_K^*$. 
\begin{lemma}\label{lemma_nondec}
For $C_1 \leq \dots \leq C_K $ and $m_{\text{tot}} \leq 1$, the objective function of (\ref{eqn_opt3}) satisfies $\Theta_{\mathfrak A, \mathfrak D}(\bm m) \leq \Theta_{\mathfrak A, \mathfrak D}( \tilde{ \bm m})$, where $m_i= \tilde m_i$, for $i \in [K]\setminus\{r,r+1\}$, and some $r \in [K \! - \! 1]$. Additionally, $m_r= \tilde m_{r+1}= \alpha + \delta$, $m_{r+1}= \tilde m_{r}= \alpha $, for $\delta, \alpha \geq 0$, and $m_1 \geq m_2 \geq \dots \geq m_{r}$. 
\end{lemma}
\vspace{-0.2in}
\begin{proof}
For $\bm m =[m_1, m_2, \dots, m_{r-1}, \alpha\!+\! \delta, \alpha, m_{r+2}, \dots, m_K]$ and $\tilde{ \bm  m} =[m_1, m_2, \dots, m_{r-1}, \alpha ,\alpha\!+\! \delta, m_{r+2}, \dots, m_K]$, we have $ \Theta_{\mathfrak A, \mathfrak D}(\bm m) -\Theta_{\mathfrak A, \mathfrak D}( \tilde{ \bm m}) = \chi_{1}+\chi_{2}$, where 
\begin{align}
\chi_{1}&=\frac{1-m_r}{C_r}+\frac{1-m_{r+1}}{C_{r+1}}-\frac{1-\tilde m_r}{C_r}-\frac{1-\tilde m_{r+1}}{C_{r+1}} \ = \delta \left( \frac{1}{C_{r+1}}-\frac{1}{C_{r}}\right), \\
\chi_{2}&= \! \sum_{i=1}^{r-1} \! \left( \! \frac{\min\{m_i,\tilde m_r\}}{C_r}+\frac{\min\{m_i,\tilde m_{r+1}\}}{C_{r+1}} \! \right)  \! - \! \sum_{i=1}^{r-1} \! \left(\frac{\min\{m_i,m_r\}}{C_r} + \frac{\min\{m_i, m_{r+1}\}}{C_{r+1}} \right), \\
&=\left(\frac{1}{C_{r+1}}-\frac{1}{C_{r}} \right) \sum_{i=1}^{r-1} \left( \min\{m_i,\alpha\!+\! \delta\} - \min\{m_i,\alpha\}\right) \ = \delta (r-1) \left( \frac{1}{C_{r+1}}-\frac{1}{C_{r}}\right).
\end{align}
Thus, $\Theta_{\mathfrak A, \mathfrak D}(\bm m)\! -\Theta_{\mathfrak A, \mathfrak D}( \tilde{ \bm m})\! = r \delta  \left( \frac{1}{C_{r+1}}\!-\!\frac{1}{C_{r}}\right)\! \leq 0$, as $C_{r+1}\! \geq C_{r} $.
\end{proof}
\vspace{-.05 in}
Using Lemma \ref{lemma_nondec}, (\ref{eqn_opt3}) can be simplified to (\ref{eqn_opt4}).
\vspace{-.05 in}
\begin{lemma}
For $C_1 \leq \dots \leq C_K $ and $m_{\text{tot}} \leq 1$, optimization problem (\ref{eqn_opt2}) reduces to
\begin{subequations} \label{eqn_opt4}
\begin{align}
& \min_{ \bm m}  
& & \sum_{k=1}^{K} \dfrac{1-k \ \! m_k}{C_k} \\
& \ \text{s.t.}
& &  \sum_{k=1}^{K} m_k \leq m_{\text{tot}}, \ \! 0 \leq m_{k+1}  \leq m_{k}, \ \forall \ \! k \in [K \! - \! 1].
\end{align}
\end{subequations}
\end{lemma}
Equivalently, the optimal memory allocation for (\ref{eqn_opt4}) is obtained by solving %the following maximization problem.
\begin{subequations} \label{eqn_app_mem_opt1}
\begin{align}
& \max_{ \bm m}  
& & \sum_{k=1}^{K} \dfrac{k \ \! m_k}{C_k} \\
& \ \text{s.t.}
& &  \sum_{k=1}^{K} m_k \leq m_{\text{tot}}, \ \! 0 \leq m_{k+1}  \leq m_{k}, \ \forall \ \! k \in [K \! - \! 1].
\end{align}
\end{subequations}

%\begin{subequations} \label{eqn_app_mem_opt1}
%\begin{align}
%& \max_{ \bm m \succeq 0}  
%& & \sum_{k=1}^{K} \dfrac{k \ \! m_k}{C_k} \\
%& \text{subject to}
%& &  \! \! \! \! \! \! \! \begin{bmatrix}
% 1 & 1 & 1 & \cdots \! & 1 &  1 \\
% \!-1 &  1 &  0 & \cdots \!&  0 & 0 \\
%0 &  \! \!\!-1 &   1 & \cdots \! &  0 &  0 \\
% \vdots &  \vdots &   \! & \vdots &  \vdots &  \vdots \\
%0 &  0 & 0 & 0    \!  &  \!\!\!-1 &  1
%\end{bmatrix}  \begin{bmatrix}
% m_1  \\
% m_2  \\
%m_3  \\
% \vdots \\
%m_K 
%\end{bmatrix}  \leq  \begin{bmatrix}
% m_{\text{tot}}  \\
% 0  \\
%0  \\
% \vdots \\
%0 
%\end{bmatrix}.
%\end{align}
%\end{subequations}
Finally, the optimal memory allocation in Theorem \ref{thm_optmem} is obtained by solving the dual of the linear program in (\ref{eqn_app_mem_opt1}).%, which is given by
%\begin{subequations} \label{eqn_app_mem_opt2}
%\begin{align}
%& \min_{ \bm y \succeq 0}  
%& &  m_{\text{tot}} \ \! y_1\\
%& \text{subject to}
%& &  \! \! \! \! \! \! \! \begin{bmatrix}
% 1 & \!\!-1 & 0  & 0  & \cdots  & 0  &  0 \\
% 1 &  1 & \!\!-1 & 0  & \cdots  & 0  &  0 \\
% 1 &  0 &  1 & \!\! -1 & \cdots  &  0 &  0 \\
%\vdots &  \vdots &  \vdots & \vdots &  &  \vdots &  \vdots \\
% 1 & 0 & 0  & 0  & \cdots  & 0  &  1 
%\end{bmatrix}  \begin{bmatrix}
% y_1  \\
% y_2  \\
%y_3  \\
% \vdots \\
%y_K 
%\end{bmatrix}  \geq  \begin{bmatrix}
% 1/C_1  \\
% 2/C_2  \\
% 3/C_3  \\
% \vdots \\
% K/C_K 
%\end{bmatrix}.
%\end{align}
%\end{subequations}
% ################################################################################# 
%\vspace{-0.2in}
%\newpage
\bibliographystyle{IEEEtran}
\bibliography{IEEEabrv,references}
 
\end{document}